\documentclass[showpacs,showkeys,preprintnumbers,amsmath,amssymb,aip,groupedaddress]{revtex4-1}

\usepackage{graphicx}       
\usepackage{psfrag}
\usepackage{tikz}
\usepackage{mathptmx}
\usepackage{amsmath,amscd, amsthm}
\usetikzlibrary{arrows,shapes,backgrounds}

\usepackage{ifthen}
\newboolean{includefigs}
\setboolean{includefigs}{true}
\newboolean{psfigs}
\setboolean{psfigs}{true}
\newboolean{tikzfigs}
\setboolean{tikzfigs}{false}

\newcommand{\condcomment}[2]{\ifthenelse{#1}{#2}{}}
\def\celine#1#2{\textbf{#2}}

\newtheorem{thm}{Theorem}[section]
\newtheorem{cor}[thm]{Corollary}
\newtheorem{lem}[thm]{Lemma}

\newtheorem{defn}[thm]{Definition}

\begin{document}

\title{Affine extensions of non-crystallographic Coxeter groups induced by projection}

\author{Pierre-Philippe Dechant}
\email{pierre-philippe.dechant@durham.ac.uk}
\affiliation{Ogden Centre for Fundamental Physics, Department of Physics, University of Durham, South Rd, Durham DH1 3LE, UK and \\Department of Mathematics, York Centre for Complex Systems Analysis, University of York, Heslington, York, UK} 

\author{C\'eline B\oe hm}
\email{c.m.boehm@durham.ac.uk}
\affiliation{Ogden Centre for Fundamental Physics, Department of Physics, University of Durham, South Rd, Durham DH1 3LE, UK}

\author{Reidun Twarock} \email{rt507@york.ac.uk}
\affiliation{Department of Mathematics, York Centre for Complex Systems Analysis, University of York, Heslington, York, UK}

\preprint{IPPP/11/65, DCPT/11/130}

\date{\today}

\begin{abstract}
	In this paper, we show that affine extensions of non-crystallographic Coxeter groups can be derived via Coxeter-Dynkin diagram foldings and projections of affine extended versions of the root systems   $E_8$, $D_6$ and $A_4$. We show that the induced affine extensions of the non-crystallographic groups $H_4$, $H_3$ and $H_2$ correspond to a distinguished subset of the Kac-Moody-type extensions considered in  \cite{DechantTwarockBoehm2011H3aff}. This class of extensions was  motivated by physical applications in icosahedral systems in biology (viruses), physics (quasicrystals) and chemistry (fullerenes). By connecting these here to extensions of $E_8$, $D_6$ and $A_4$, we place them into the broader context of crystallographic lattices such as $E_8$, suggesting their  potential for applications in high energy physics, integrable systems and modular form theory. 
	By inverting the projection, we make the case for admitting
	different number fields in the Cartan matrix, which could open up enticing
	possibilities in hyperbolic geometry and rational conformal field theory.
\end{abstract}

 

\maketitle

\section{Introduction}

The classification of finite-dimensional simple Lie algebras by Cartan and Killing is one of the mile stones of modern mathematics. The study of these algebras is essentially reduced to that of root systems and their Weyl groups, and all their geometric content is contained in  Cartan matrices and visualised in Dynkin diagrams. The problem ultimately amounts to classifying all possible Cartan matrices \cite{FuchsSchweigert1997}. 

Coxeter groups  describe (generalised) reflections \cite{Coxeter1934discretegroups}, and thus encompass the above Weyl groups, which are the reflective symmetry groups of the relevant root systems. In fact, the finite Coxeter groups are precisely the finite Euclidean reflection groups \cite{Coxeter1935Enumeration}. However, since the root systems arising in Lie Theory are related to lattices, the Weyl groups are automatically crystallographic in  nature. Non-crystallographic   Coxeter groups, i.e. those that do not stabilise any lattice (in the dimension equal to their rank), therefore cannot arise in the Lie Theory context, and as a consequence, they have not been studied as intensely.  
They include the groups $H_2$, $H_3$ and the largest non-crystallographic group $H_4$;   the icosahedral group $H_3$ and its rotational subgroup $I$ are of particular practical importance as $H_3$ is the largest discrete symmetry group of physical space. Thus, many 3-dimensional systems with  `maximal symmetry', like  viruses in biology \cite{Stockley2010emerging, Caspar:1962, Twarock:2006b, Janner:2006b, Zandi:2004}, fullerenes in chemistry \cite{Kroto:1985,Kroto:1992, Twarock:2002b, Kustov:2008} and quasicrystals in physics \cite{Katz:1989, Senechal:1996, MoodyPatera:1993b, Levitov:1988}, can be modeled using Coxeter groups. 

Affine Lie algebras have also been studied for a long time, and many of the salient features of the theory of simple Lie algebras carry over to the affine case. More recently, Kac-Moody Theory has provided another  framework in which generalised Cartan matrices induce interesting algebraic structures that preserve many of the  features encountered in the simple and affine cases \cite{Kac1994InfDimLA}. However, such considerations again only give rise to extensions of crystallographic Coxeter groups. These infinite Coxeter groups are usually constructed directly from the finite Coxeter groups by introducing affine reflection planes (planes not containing the origin). While these infinite {counterparts} to the crystallographic Coxeter groups have  been intensely studied \cite{Humphreys1990Coxeter}, much less is known about their non-crystallographic {counterparts} \cite{Twarock:2002a}. 
Recently, we have derived novel affine extensions of the non-crystallographic Coxeter groups $H_2$, $H_3$ and $H_4$ in  two, three and four dimensions, based on an extension of their Cartan matrices following the Kac-Moody formalism in Lie Theory  \cite{DechantTwarockBoehm2011H3aff}. 

In this paper, we develop a different approach and induce such affine extensions of the non-crystallographic groups $H_2$, $H_3$ and $H_4$ from affine extensions of the crystallographic groups $A_4$, $D_6$ and $E_8$, via  projection from the higher-dimensional setting. 
 Specifically, there exists a projection from the  root system of $E_8$, the largest exceptional Lie algebra, to the  root system of $H_4$,  the largest non-crystallographic Coxeter group \cite{Humphreys1990Coxeter}, and,  due to the inclusions $A_4\subset D_6\subset E_8$ and $H_2\subset H_3 \subset H_4$,  also corresponding projections for the other non-crystallographic Coxeter groups.  

We apply these projections here to the extended root systems of the groups $A_4$, $D_6$ and $E_8$. As  expected,  extending by a single node recovers only those affine extensions known in the literature. However, we also consider simply-laced extensions with two additional nodes in the Kac-Moody formalism, and consider their compatibility with the projection formalism.
Specifically, we use the projection of the affine root as an affine root for the projected root system, and
thereby find a distinguished subset of the solutions in the classification scheme presented in   \cite{DechantTwarockBoehm2011H3aff}.

The $E_8$ root system, and the related structures:   the $E_8$ lattice, the Coxeter group, the Lie algebra and the Lie group, are `exceptional' structures, and are  of critical importance in mathematics and in theoretical physics \cite{FuchsSchweigert1997}. For instance, they occur in the context of Lie algebras, simple group theory and modular form theory, as well as lattice packing theory \cite{Conway1987Packings, Torquato2009DensePackings, KeefDechantTwarock2012Packings}.
In theoretical physics, $E_8$ is central to String Theory, as it is the gauge group for the $E_8 \times E_8$ heterotic string \cite{Gross1985HE}. More recently, via the Ho{\v r}ava-Witten picture \cite{HoravaWitten1996HW, Witten1996CY, HoravaWitten1996} and other developments \cite{Nicolai1994E10, DamourHenneauxNicolai2002E10,DamourHenneauxNicolai2002Billiards, HennauxPersson2008SpacelikeSingularitiesAndHiddenSymmetriesofGravity,
HenneauxPerssonWesley2008CoxeterGroupStructure},  $E_8$ and its affine extensions and overextensions (e.g. $E_8^+$ and $E_8^{++}$) have emerged as the most likely candidates  for the underlying symmetry of M-Theory. It is also fundamental in the context of Grand Unified Theories \cite{GeorgiGlashow1974GUT,Bars1980E8GUT, HeckmanVafa2010E8FGUT}, as it is the largest irreducible group that can accommodate the Standard Model gauge group $SU(3)\times SU(2)\times U(1)$. Our new link between  affine extensions of crystallographic Coxeter groups such as $E_8^+$ and their non-crystallographic counterparts could thus turn out to be important in High Energy Physics, e.g. in String Theory or in possible extensions of the Standard Model above the TeV scale after null findings at the LHC.

The structure of this paper is as follows. Section \ref{MP} reviews some standard  results to provide the necessary background  for our novel construction. Section \ref{sec_Cox} discusses the basics of Coxeter groups. Section \ref{sec_proj} introduces the relationship between $E_8$ and $H_4$, and discusses how it manifests itself on the level of the root systems, the representation theory, and the Dynkin diagram foldings and projection formalism. Section \ref{sec_std_single} introduces  affine extensions of crystallographic Coxeter groups, and   presents the standard affine extensions of the groups relevant in our context. 
In Section \ref{sec_proj_aff}, we  compute where the affine roots of the standard extensions of the crystallographic groups  map under the projection formalism and examine the resulting induced affine extensions of the non-crystallographic groups. Section \ref{sec_autos} discusses Coxeter-Dynkin diagram automorphisms of the simple and affine  groups, and shows that the induced affine extensions are invariant under these automorphisms. In Section \ref{sec_further}, we consider  affine extending the crystallographic groups  by two nodes and show that these do not induce any further affine extensions. In Section \ref{sec_sum}, we briefly review the novel Kac-Moody-type extensions of non-crystallographic Coxeter groups from a recent paper and compare the induced extensions with the classification scheme presented there (Section \ref{sec_class}). In Section \ref{sec_concl}, {we conclude that in a wide class of extensions (single extensions or simply-laced double extensions with trivial projection kernel), the ten induced cases considered here are the only ones that are compatible with the projection}. We also discuss how lifting affine extensions of non-crystallographic groups to the crystallographic setting, as well as symmetrisability of the resulting matrices,  motivate a study of Cartan matrices over extended number fields.


\section{Mathematical Preliminaries}\label{MP}

In this section, we introduce the context of our construction, with the relevant concepts and the known links between them, as illustrated in Fig. \ref{EDAH_overview}.  
We introduce Coxeter groups and their root systems in Section \ref{sec_Cox}, and discuss how certain crystallographic and non-crystallographic groups are related via projection (Section \ref{sec_proj}). Affine extensions of the crystallographic Coxeter groups are introduced in Section \ref{sec_std_single}. Affine extensions of the non-crystallographic Coxeter groups in dimensions two, three and four  have been discussed in our previous papers  \cite{DechantTwarockBoehm2011H3aff,Twarock:2002a} (see Section \ref{sec_sum}). Here, we present a different construction of such affine extensions, by inducing them from the known affine extensions of the crystallographic Coxeter groups via projection from the higher-dimensional setting. These induced extensions will be shown to be a subset of those derived in  \cite{DechantTwarockBoehm2011H3aff}.

\condcomment{\boolean{psfigs}}{
\begin{figure}
\begin{center}

	\includegraphics[width=15cm]{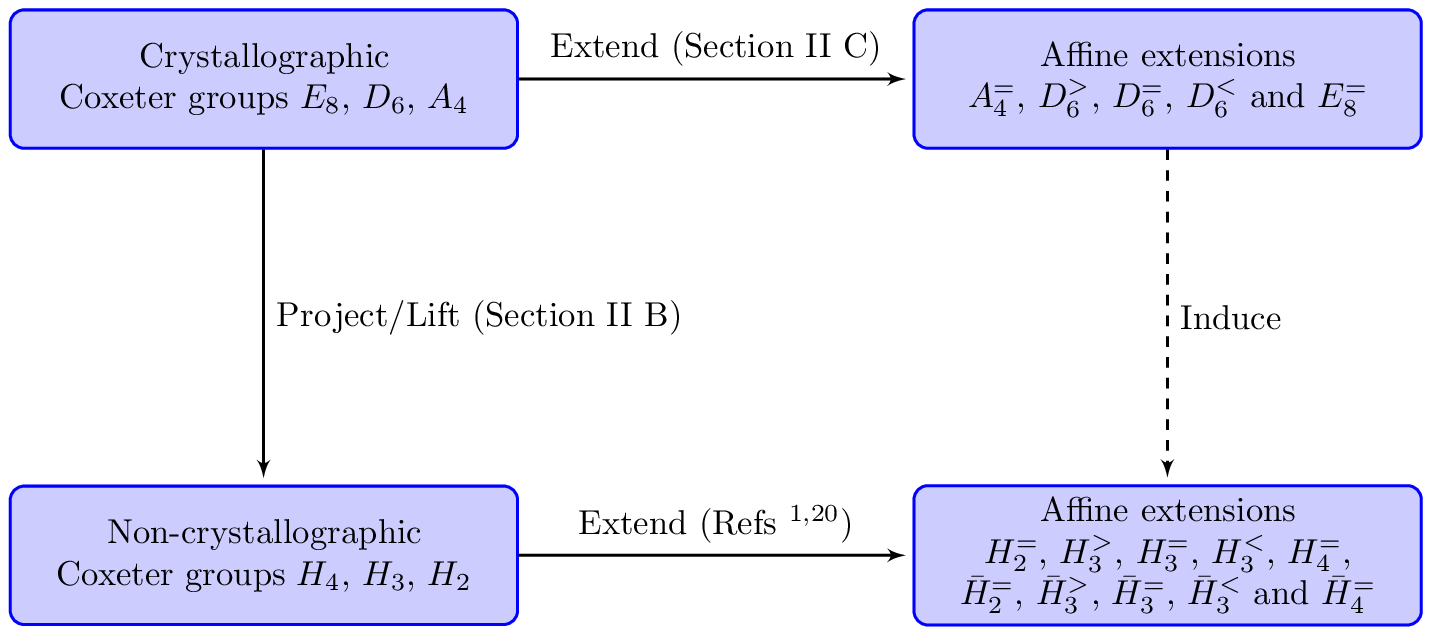}

\caption[Overview]{Context of this paper: Section \ref{sec_Cox} introduces Coxeter groups (left), and Section \ref{sec_proj} discusses how certain crystallographic and non-crystallographic groups are related via projection (left arrow). Section \ref{sec_std_single} discusses the known affine extensions of the crystallographic Coxeter groups (upper arrow), and affine extensions of non-crystallographic Coxeter groups have been discussed in  \cite{DechantTwarockBoehm2011H3aff,Twarock:2002a} (lower arrow). In this paper, we present a novel way of  inducing affine extensions of the non-crystallographic groups via projection from the affine extensions of the crystallographic groups  (dashed arrow on the right), yielding a distinguished subset of those derived in  \cite{DechantTwarockBoehm2011H3aff}.}
\label{EDAH_overview}
\end{center}
\end{figure}}

\celine{
\begin{center}
\begin{figure}
\begin{tikzpicture}
[auto,
decision/.style={diamond, draw=blue, thick, fill=blue!20,
text width=4.5em, text badly centered,
inner sep=1pt},
block/.style ={rectangle, draw=blue, thick, fill=blue!20,
text width=14em, text centered, rounded corners,
minimum height=4em},
line/.style ={draw, thick, -latex',shorten >=2pt},
cloud/.style ={draw=red, thick, ellipse,fill=red!20,
minimum height=2em}]
\matrix [column sep=20mm,row sep=17mm]
{
 \node [block] (C) {Crystallographic\\ Coxeter groups $E_8$, $D_6$, $A_4$}; & &\node [block] (Caff) {Affine extensions \\$A_4^=$, $D_6^>$, $D_6^=$, $D_6^<$ and $E_8^=$}; \\
 \\
\node [block] (NC) {Non-crystallographic \\ Coxeter groups $H_4$, $H_3$, $H_2$}; & & \node [block] (NCaff) {Affine extensions \\$H_2^=$, $H_3^>$, $H_3^=$, $H_3^<$ and $H_4^=$}; \\
};
\begin{scope}[every path/.style=line]
\path (C)  to  node [midway]   {Project/Lift (Section \ref{sec_proj})} (NC);
\path (C) to   node [midway] {Extend (Section \ref{sec_std_single})} (Caff);
\path (Caff) to [dashed]  node [midway] {Induce} (NCaff);
\path (NC) to   node [midway] {Extend ( \cite{DechantTwarockBoehm2011H3aff,Twarock:2002a})} (NCaff);
\end{scope}
\end{tikzpicture}
\caption[Overview]{Context of this paper: Section \ref{sec_Cox} introduces Coxeter groups (left), and Section \ref{sec_proj} discusses how certain crystallographic and non-crystallographic groups are related via projection (left arrow). Section \ref{sec_std_single} discusses the known affine extensions of the crystallographic Coxeter groups (upper arrow), and affine extensions of non-crystallographic Coxeter groups have been discussed in  \cite{DechantTwarockBoehm2011H3aff,Twarock:2002a} (lower arrow). In this paper, we  induce affine extensions of the non-crystallographic groups in a novel way from the affine extensions of the crystallographic groups via projection (dashed arrow on the right). }
\label{EDAH_overview}
\end{figure}
\end{center}}{}

\subsection{Finite Coxeter groups and root systems} \label{sec_Cox}
\begin{defn}[Coxeter group] A \emph{Coxeter group} is a group generated by some involutive generators $s_i, s_j \in S$ subject to relations of the form $(s_is_j)^{m_{ij}}=1$ with $m_{ij}=m_{ji}\ge 2$ for $i\ne j$. The matrix $A$ with entries $A_{ij}=m_{ij}$ is called the \emph{Coxeter matrix}. 
\end{defn}

The  finite Coxeter groups have a geometric representation where the involutions are realised as reflections at hyperplanes through the origin in a Euclidean vector space $\mathcal{E}$. In particular, let $(\cdot \vert \cdot)$ denote the inner product in $\mathcal{E}$, and $\lambda$, $\alpha\in\mathcal{E}$.   
\begin{defn}[Reflection] 
The generator $s_\alpha$ corresponds to the \emph{reflection}
\begin{equation}\label{reflect}
s_\alpha: \lambda\rightarrow s_\alpha(\lambda)=\lambda - 2\frac{(\lambda\vert\alpha)}{(\alpha\vert\alpha)}\alpha
\end{equation}
 at a hyperplane perpendicular to the  \emph{root vector} $\alpha$. 
\end{defn}

The action of the Coxeter group is  to permute these root vectors, and its  structure is thus encoded in the collection  $\Phi\in \mathcal{E}$ of all such root vectors,  the root system: 
\begin{defn}[Root system] 
A \emph{root system} $\Phi$ is a finite set of non-zero vectors in $\mathcal{E}$ such that  the following two conditions hold: 
\begin{enumerate}
\item $\Phi$ only contains a root $\alpha$ and its negative, but no other scalar multiples: $\Phi \cap \mathbb{R}\alpha=\{-\alpha, \alpha\}\,\,\,\,\forall\,\, \alpha \in \Phi$. 
\item $\Phi$ is invariant under all reflections corresponding to vectors in $\Phi$: $s_\alpha\Phi=\Phi \,\,\,\forall\,\, \alpha\in\Phi$.
\end{enumerate}
\end{defn}
For a crystallographic Coxeter group, a subset $\Delta$ of $\Phi$, called \emph{simple roots}, is sufficient to express every element of $\Phi$ via a $\mathbb{Z}$-linear combination with coefficients of the same sign. 
$\Phi$ is therefore  completely characterised by this basis of simple roots, which in turn completely characterises the Coxeter group. In the case of the non-crystallographic Coxeter groups $H_2$, $H_3$ and $H_4$, the same holds for the extended integer ring $\mathbb{Z}[\tau]=\lbrace a+\tau b| a,b \in \mathbb{Z}\rbrace$, where $\tau$ is   the golden ratio $\tau=\frac{1}{2}(1+\sqrt{5})$. Note that together with its  Galois conjugate $\tau'\equiv \sigma=\frac{1}{2}(1-\sqrt{5})$, $\tau$ satisfies the quadratic equation $x^2=x+1$.  
In the following, we will call the exchange of $\tau$ and $\sigma$  \emph{Galois conjugation}, and denote it by $x\rightarrow \bar{x}=x(\tau\leftrightarrow \sigma)$.

The structure of the set of simple roots is encoded in the Cartan matrix, which contains the geometrically invariant information of the root system as follows: 
\begin{defn}[Cartan matrix and Coxeter-Dynkin diagram] The  \emph{Cartan matrix} of a set of simple roots $\alpha_i\in\Delta$ is defined as the matrix
\begin{equation}\label{CM}
	A_{ij}=2\frac{(\alpha_i\vert \alpha_j)}{(\alpha_i\vert \alpha_i)}.
\end{equation}
A graphical representation of the geometric content is given by \emph{Coxeter-Dynkin diagrams}, in which nodes correspond to simple roots, orthogonal roots are not connected, roots at $\frac{\pi}{3}$ have a simple link, and other angles $\frac{\pi}{m}$ have a link with a label $m$. 
\end{defn}

 Note that  Cartan matrix entries of $\tau$ and $\sigma$ yield Coxeter diagram labels of $5$ and $\frac{5}{2}$, respectively,  since in the simply-laced setting $A_{ij}=-2\cos\frac{\pi}{m_{ij}}$, ${\tau}=2\cos \frac{\pi}{5}$ and ${-\sigma}=2\cos \frac{2\pi}{5}$.  Such fractional values can also be understood as angles in hyperbolic space \cite{Coxeter1973regular}. By the crystallographic restriction theorem, there are no lattices (i.e. periodic structures) with such non-crystallographic symmetry $H_2$, $H_3$ and $H_4$ in two, three and four dimensions, respectively.
For these non-crystallographic Coxeter groups one therefore needs to move from a lattice to a quasilattice setting.

\subsection{From $E_8$ to $H_4$: standard Dynkin diagram foldings and projections}\label{sec_proj}
The largest exceptional (crystallographic) Coxeter group $E_8$ and the largest non-crystallographic Coxeter group $H_4$ are closely related. 
This connection between $E_8$ and $H_4$ can be exhibited in various ways, including Coxeter-Dynkin diagram foldings in the Coxeter group picture \cite{Shcherbak:1988}, relating the root systems \cite{MoodyPatera:1993b, Koca:1998, Koca:2001}, and in terms of the representation theory \cite{Shcherbak:1988, MoodyPatera:1993b, Koca:1998, Koca:2001, Katz:1989}. For illustrative purposes, we focus on the folding picture first. 

Following \cite{Shcherbak:1988}, we consider the Dynkin diagram of $E_8$ (top left of Fig. \ref{figE8}), where  we have  labeled the simple roots $\alpha_1$ to  $\alpha_8$. We fold the diagram suggestively (bottom left of Fig. \ref{figE8}), and define the 
combinations $s_{\beta_1}=s_{\alpha_1} s_{\alpha_7}$,   $s_{\beta_2}=s_{\alpha_2} s_{\alpha_6}$, $s_{\beta_3}=s_{\alpha_3} s_{\alpha_5}$ and   $s_{\beta_4}=s_{\alpha_4} s_{\alpha_8}$. It can be shown that the subgroup with the generators $\beta_i$ is in fact isomorphic to $H_4$ (top right) \cite{Bourbaki1981Lie, Shcherbak:1988}. This amounts to demanding that the simple roots of $E_8$ project onto the simple roots of $H_4$ and their $\tau$-multiples, as denoted on the bottom right of Fig. \ref{figE8}. One choice of simple roots for $H_4$ is  $a_1=\frac{1}{2}(-\sigma, -\tau, 0, -1)$, $a_2=\frac{1}{2}(0, -\sigma, -\tau,  1)$, $a_3=\frac{1}{2}(0,1, -\sigma, -\tau)$ and $a_4=\frac{1}{2}(0, -1, -\sigma, \tau)$, and in that case the highest root is  $\alpha_H=(1,0,0,0)$.  In the bases of simple roots $\alpha_i$ and $a_i$, the projection is given by
\begin{equation}
	\pi_\parallel = \begin{pmatrix}
	    1&0&0&0&0&0&\tau&0
	\\   0&1&0&0&0&\tau&0&0
	\\  0&0&1&0&\tau&0&0&0
	\\  0&0&0&\tau&0&0&0&1
	 \end{pmatrix}. \label{projpi}
\end{equation}

\condcomment{\boolean{includefigs}}{
\condcomment{\boolean{tikzfigs}}{
\begin{figure}
	\begin{center}
       \begin{tabular}{@{}c@{ }c@{ }}
		\begin{tikzpicture}[
		    knoten/.style={        circle,      inner sep=.15cm,        draw}
		   ]

		  \node at  (1,.5) (knoten1) [knoten,  color=white!0!black] {};
		  \node at  (2.5,.5) (knoten2) [knoten,  color=white!0!black] {};
		  \node at  (4,.5) (knoten3) [knoten,  color=white!0!black] {};
		  \node at  (5.5,.5) (knoten4) [knoten,  color=white!0!black] {};

		  \node at  (7,.5) (knoten5) [knoten,  color=white!0!black] {};
		  \node at (8.5,.5) (knoten6) [knoten,  color=white!0!black] {};
		  \node at (10,.5) (knoten7) [knoten,  color=white!0!black] {};
		  \node at (7,2.0) (knoten8) [knoten,  color=white!0!black] {};

		\node at  (1,0)  (alpha1) {$\alpha_1$};
		\node at  (2.5,0)  (alpha2) {$\alpha_2$};
		\node at  (4,0)  (alpha3) {$\alpha_3$};
		\node at  (5.5,0)  (alpha4) {$\alpha_4$};
		\node at  (7,0)  (alpha7) {$\alpha_5$};
		\node at (8.5,0)  (alpha6) {$\alpha_6$};
		\node at (10,0)  (alpha5) {$\alpha_7$};
		\node at (7,2.5) (alpha8) {$\alpha_8$};

		  \path  (knoten1) edge (knoten2);
		  \path  (knoten2) edge (knoten3);
		  \path  (knoten3) edge (knoten4);
		  \path  (knoten4) edge (knoten5);
		  \path  (knoten5) edge (knoten6);
		  \path  (knoten6) edge (knoten7);
		  \path  (knoten5) edge (knoten8);

		\end{tikzpicture}& \hspace{0.3 cm} 
 		\begin{tikzpicture}[
		    knoten/.style={        circle,      inner sep=.15cm,        draw}  
		   ]

		  \node at (1,0.5) (knoten1) [knoten,  color=white!0!black] {};
		  \node at (3,0.5) (knoten2) [knoten,  color=white!0!black] {};
		  \node at (5,0.5) (knoten3) [knoten,  color=white!0!black] {};
		  \node at (7,0.5) (knoten4) [knoten,  color=white!0!black] {};

		\node at (1,0)  (a1) {$a_1$};
		\node at (3,0)  (a2) {$a_2$};
		\node at (5,0)  (a3) {$a_3$};
		\node at (6,0.75)  (tau) {$5$};
		\node at (7,0)  (a4) {$a_4$};

		  \path  (knoten1) edge (knoten2);
		  \path  (knoten2) edge (knoten3);
		  \path  (knoten3) edge (knoten4);

		\end{tikzpicture}	

\\
\vspace{0.5cm}\\
    \large $\Downarrow$ \small{fold} &  \large $\Uparrow$  \small{$\mathbb{Z}\rightarrow \mathbb{Z}[\tau]$}\\
\vspace{0.5cm}\\
	
\begin{tikzpicture}[
    knoten/.style={        circle,      inner sep=.15cm,        draw}
   ]
  
  \node at (1,1.5) (knoten1) [knoten,  color=white!0!black] {};
  \node at (3,1.5) (knoten2) [knoten,  color=white!0!black] {};
  \node at (5,1.5) (knoten3) [knoten,  color=white!0!black] {};
  \node at (7,1.5) (knoten4) [knoten,  color=white!0!black] {};

  \node at (1,.5) (knoten5) [knoten,  color=white!0!black] {};
  \node at (3,.5) (knoten6) [knoten,  color=white!0!black] {};
  \node at (5,.5) (knoten7) [knoten,  color=white!0!black] {};
  \node at (7,.5) (knoten8) [knoten,  color=white!0!black] {};

\node at (1,2)  (alpha1) {$\alpha_1$};
\node at (3,2)  (alpha2) {$\alpha_2$};
\node at (5,2)  (alpha3) {$\alpha_3$};
\node at (7,2)  (alpha4) {$\alpha_4$};
\node at (1,0)  (alpha7) {$\alpha_7$};
\node at (3,0)  (alpha6) {$\alpha_6$};
\node at (5,0)  (alpha5) {$\alpha_5$};
\node at (7,0)  (alpha8) {$\alpha_8$};
\node at (9,1.4)  (pi) {$\pi_\parallel$};
\node at (9,0.9)  (ra) {$\Rightarrow$};
\node at (9,0.4)  (rp) {$\text{project}$};

  \path  (knoten1) edge (knoten2);
  \path  (knoten2) edge (knoten3);
  \path  (knoten3) edge (knoten4);
  \path  (knoten5) edge (knoten6);
  \path  (knoten6) edge (knoten7);
  \path  (knoten7) edge (knoten8);
  \path  (knoten4) edge (knoten7);
 
\end{tikzpicture}
	 \large  &
	\begin{tikzpicture}[
	    knoten/.style={        circle,      inner sep=.15cm,        draw}
	   ]

	  \node at (1,1.5) (knoten1) [knoten,  color=white!0!black] {};
	  \node at (3,1.5) (knoten2) [knoten,  color=white!0!black] {};
	  \node at (5,1.5) (knoten3) [knoten,  color=white!0!black] {};
	  \node at (7,1.5) (knoten4) [knoten,  color=white!0!black] {};

	  \node at (1,.5) (knoten5) [knoten,  color=white!0!black] {};
	  \node at (3,.5) (knoten6) [knoten,  color=white!0!black] {};
	  \node at (5,.5) (knoten7) [knoten,  color=white!0!black] {};
	  \node at (7,.5) (knoten8) [knoten,  color=white!0!black] {};

	\node at (1,2)  (alpha1) {$a_1$};
	\node at (3,2)  (alpha2) {$a_2$};
	\node at (5,2)  (alpha3) {$a_3$};
	\node at (7,2)  (alpha4) {$\tau a_4$};
	\node at (1,0)  (alpha7) {$\tau a_1$};
	\node at (3,0)  (alpha6) {$\tau a_2$};
	\node at (5,0)  (alpha5) {$\tau a_3$};
	\node at (7,0)  (alpha8) {$a_4$};

	  \path  (knoten1) edge (knoten2);
	  \path  (knoten2) edge (knoten3);
	  \path  (knoten3) edge (knoten4);
	  \path  (knoten5) edge (knoten6);
	  \path  (knoten6) edge (knoten7);
	  \path  (knoten7) edge (knoten8);
	  \path  (knoten4) edge (knoten7);

	\end{tikzpicture}
  \\
  \end{tabular}
  \caption[$E_8$]{Coxeter-Dynkin diagram folding and projection from $E_8$ to $H_4$. The nodes correspond to simple roots and links labeled $m$ encode an angle of $\frac{\pi}{m}$ between the root vectors, with $m$ omitted if the angle is $\frac{\pi}{3}$ and no link shown if $\frac{\pi}{2}$. Note that deleting nodes $\alpha_1$ and $\alpha_7$ yields corresponding results for $D_6\rightarrow H_3$, and likewise for  $A_4\rightarrow H_2$ by further removing $\alpha_2$ and $\alpha_6$.}
\label{figE8}
\end{center}
\end{figure}}

\condcomment{\boolean{psfigs}}{
\begin{figure}
	\begin{center}
\includegraphics[width=16cm]{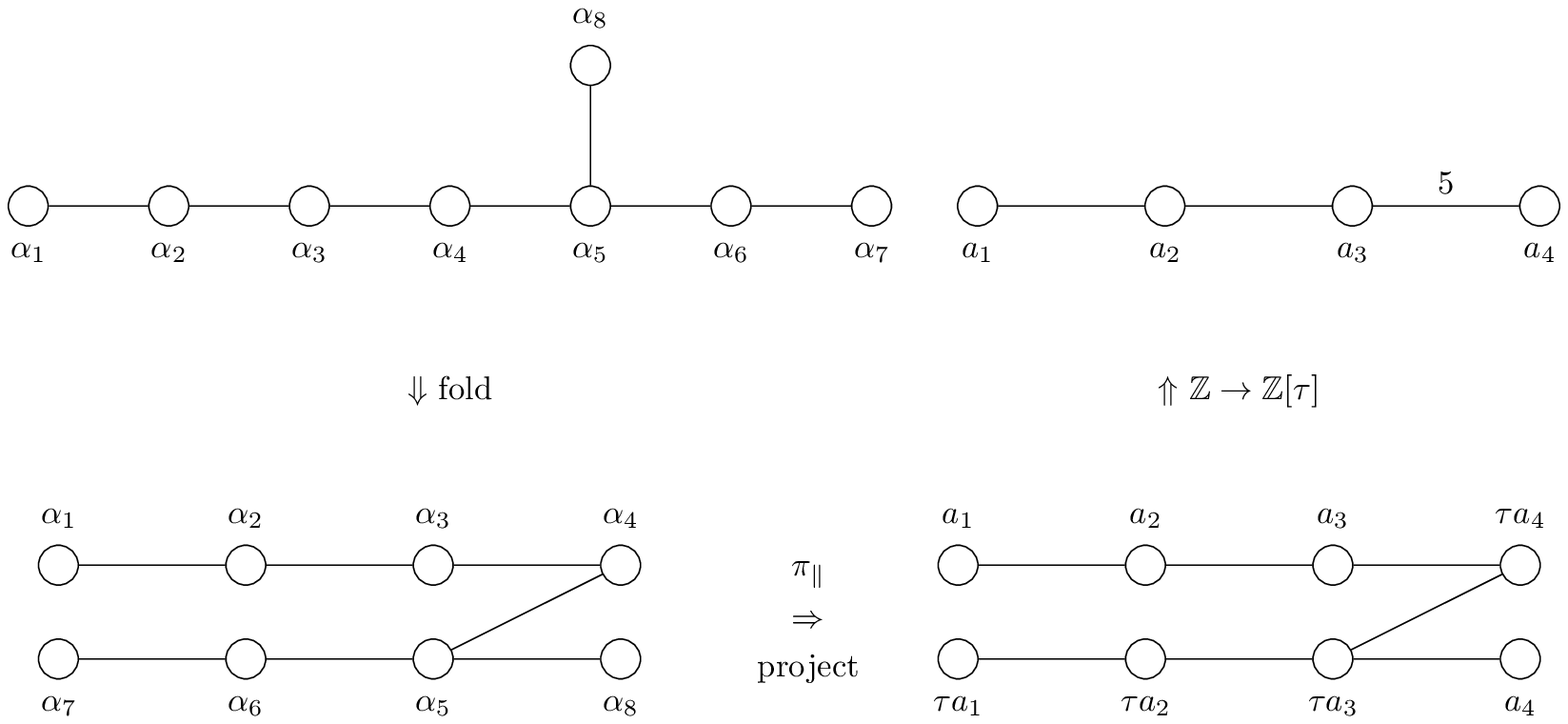}
  \caption[$E_8$]{Coxeter-Dynkin diagram folding and projection from $E_8$ to $H_4$. The nodes correspond to simple roots and links labeled $m$ encode an angle of $\frac{\pi}{m}$ between the root vectors, with $m$ omitted if the angle is $\frac{\pi}{3}$ and no link shown if $\frac{\pi}{2}$. Note that deleting nodes $\alpha_1$ and $\alpha_7$ yields corresponding results for $D_6\rightarrow H_3$, and likewise for  $A_4\rightarrow H_2$ by further removing $\alpha_2$ and $\alpha_6$.}
\label{figE8}
\end{center}
\end{figure}

}}

There are similar diagrams for $A_4$ and $D_6$ that can be obtained from the $E_8$ diagram by deleting nodes. We display these in Fig. \ref{figD6} in order to set out our notation, as the conventional way of numbering the roots in the Dynkin diagrams differs from the natural numbering in the folding picture. The only non-trivial Coxeter relation is the one corresponding to the  5-fold rotation (e.g. the relation between $\beta_3$ and $\beta_4$ for the case of $E_8$).  The additional $\beta$-generators in the higher-dimensional cases are trivial, as they can be straightforwardly shown to satisfy the relevant Coxeter relations (corresponding to 3-fold rotations) directly from the original Coxeter relations for the  $\alpha_i$s of the larger groups.

It has been observed that the $E_8$ root vectors can be realised in terms of  unit quaternions with coefficients in  $\mathbb{Z}[\tau]$ \cite{Tits1980Quaternions, MoodyPatera:1993b, Koca:1998, Koca:2001}. Specifically, the set of 120  icosians forms a discrete group under standard quaternionic multiplication, and is  a realisation of the $H_4$ root system \cite{MoodyPatera:1993b, Dechant2012CoxGA, Dechant2012Polytopes}. 
The 240 roots of $E_8$ have been shown to be in 1-1-correspondence with the 120 icosians and their 120 $\tau$-multiples, so that, schematically, $E_8 \sim H_4+\tau H_4$ holds. 
The projection considered above therefore exhibits this mapping of the simple roots of $E_8$ onto the simple roots of $H_4$ and their $\tau$-multiples. Corresponding results hold for the other groups $D_6$ and $H_3$, as well as $A_4$ and $H_2$ by inclusion. 

\condcomment{\boolean{includefigs}}{
\condcomment{\boolean{tikzfigs}}{
\begin{figure}
	\begin{center}
       \begin{tabular}{@{}c@{ }c@{ }c@{ }}
		\begin{tikzpicture}[
		    knoten/.style={        circle,      inner sep=.15cm,        draw}
		   ]
\node at (0,0) {};
		  \node at  (1,.5) (knoten1) [knoten,  color=white!0!black] {};
		  \node at  (3,.5) (knoten2) [knoten,  color=white!0!black] {};
		  \node at  (5,.5) (knoten3) [knoten,  color=white!0!black] {};
		  \node at  (7,.5) (knoten4) [knoten,  color=white!0!black] {};
		  \node at  (9,.5) (knoten5) [knoten,  color=white!0!black] {};
		  \node at (7,2) (knoten6) [knoten,  color=white!0!black] {};

		\node at  (1,0)  (alpha1) {$\alpha_1$};
		\node at  (3,0)  (alpha2) {$\alpha_2$};
		\node at  (5,0)  (alpha3) {$\alpha_3$};
		\node at  (7,0)  (alpha4) {$\alpha_4$};
		\node at  (9,0)  (alpha5) {$\alpha_5$};
		\node at  (7,2.5)  (alpha6) {$\alpha_6$};

		  \path  (knoten1) edge (knoten2);
		  \path  (knoten2) edge (knoten3);
		  \path  (knoten3) edge (knoten4);
		  \path  (knoten4) edge (knoten5);
		  \path  (knoten4) edge (knoten6);

		\end{tikzpicture}& 
 	
			\begin{tikzpicture}[
			    knoten/.style={        circle,      inner sep=.15cm,        draw}  
			   ]
\node at (0,1.5) {};
			  \node at (1,2) (knoten1) [knoten,  color=white!0!black] {};
			  \node at (3,2) (knoten2) [knoten,  color=white!0!black] {};
			  \node at (5,2) (knoten3) [knoten,  color=white!0!black] {};

			\node at (1,1.5)  (a1) {$a_1$};
			\node at (3,1.5)  (a2) {$a_2$};
			\node at (4,2.25)  (tau) {$5$};
			\node at (5,1.5)  (a3) {$a_3$};

			  \path  (knoten1) edge (knoten2);
			  \path  (knoten2) edge (knoten3);

			\end{tikzpicture}
\\

	\begin{tikzpicture}[
	    knoten/.style={        circle,      inner sep=.15cm,        draw}
	   ]
	  \node at (1,3) {};
	  \node at (1,1.5) (knoten1) [knoten,  color=white!0!black] {};
	  \node at (3,1.5) (knoten2) [knoten,  color=white!0!black] {};
	  \node at (5,1.5) (knoten3) [knoten,  color=white!0!black] {};
	  \node at (7,1.5) (knoten4) [knoten,  color=white!0!black] {};

	\node at (1,2)  (alpha1) {$\alpha_1$};
	\node at (3,2)  (alpha2) {$\alpha_2$};
	\node at (5,2)  (alpha3) {$\alpha_3$};
	\node at (7,2)  (alpha4) {$\alpha_4$};

	  \path  (knoten1) edge (knoten2);
	  \path  (knoten2) edge (knoten3);
	  \path  (knoten3) edge (knoten4);

	\end{tikzpicture}
	& 
	\begin{tikzpicture}[
	    knoten/.style={        circle,      inner sep=.15cm,        draw}  ]

	  \node at (1,1.5) (knoten1) [knoten,  color=white!0!black] {};
	  \node at (3,1.5) (knoten2) [knoten,  color=white!0!black] {};
	\node at (1,2)  (a1) {$a_1$};
	\node at (2,1.75)  (tau) {$5$};
	\node at (3,2)  (a2) {$a_2$};
	  \path  (knoten1) edge (knoten2);

	\end{tikzpicture}   	

  \\
  \end{tabular}
  \caption[$D_6$]{We display the Dynkin diagrams for $D_6$ (top left), $A_4$ (bottom left), $H_3$ (top right) and $H_2$ (bottom right), in order to fix the labelling of the roots, such that in the following the Cartan matrices can be read more easily.}
\label{figD6}
\end{center}
\end{figure}}

\condcomment{\boolean{psfigs}}{
\begin{figure}

	\begin{center}
 		\includegraphics[width=16cm]{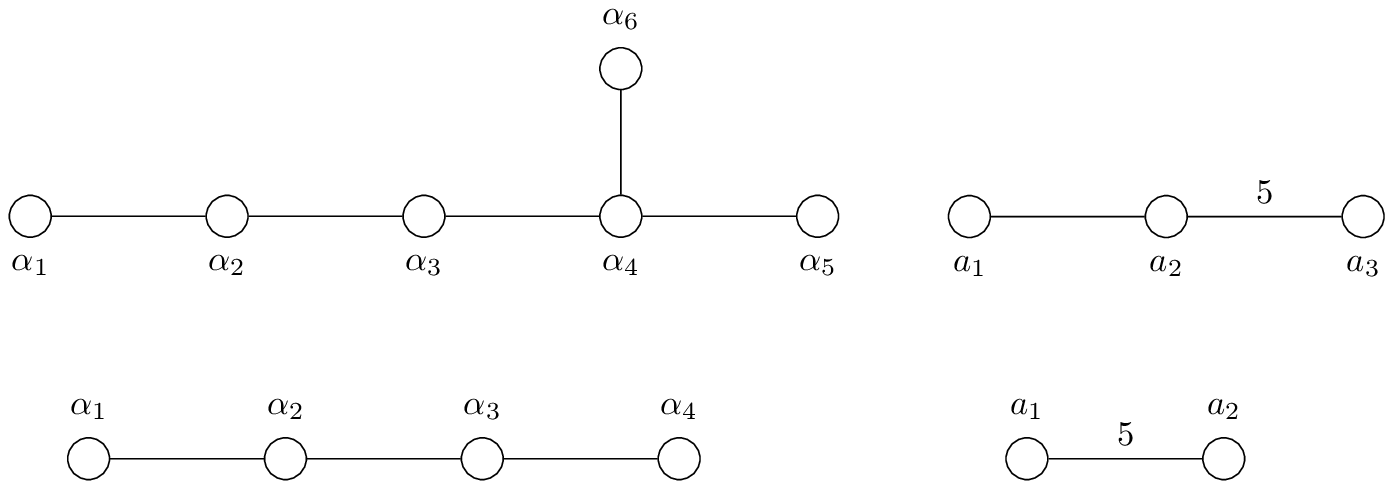}
  \caption[$D_6$]{We display the Dynkin diagrams for $D_6$ (top left), $A_4$ (bottom left), $H_3$ (top right) and $H_2$ (bottom right), in order to fix the labelling of the roots, such that in the following the Cartan matrices can be read more easily.}
\label{figD6}
\end{center}
\end{figure}}
}

From a group theoretic point of view, $E_8$ has two conjugate $H_4$-invariant subspaces. Following the terminology in  \cite{Shcherbak:1988}, we make the following definition. 
\begin{defn}[Standard and non-standard representations] We denote by the \emph{standard representation} of a Coxeter group the representation generated by mirrors forming angles $\frac{\pi}{m_{ij}}$, where $m_{ij}$ are the entries of the Coxeter matrix, i.e. $5$, $3$ and $2$ for the cases relevant here. One can also achieve a \emph{non-standard representation} of the Coxeter groups $H_i$ by instead taking mirrors at angles $2\pi/5$, i.e. by schematically going from a pentagon to a pentagram. 
\end{defn}
Note that these therefore have Coxeter diagram labels of $5$ and $\frac{5}{2}$, respectively. These are  non-equivalent irreducible representations of the non-crystallographic Coxeter groups, i.e. there is no similarity transformation that takes one to the other. Their characters are exchanged under the Galois automorphism $\tau \leftrightarrow \sigma$, and 
 the simple roots in the  non-standard representation are the Galois conjugates of the simple roots of the standard representation, for instance $a_1=\frac{1}{2}(-\sigma, -\tau, 0, -1) \leftrightarrow \bar{a}_1=\frac{1}{2}(-\tau, -\sigma, 0, -1)$. 
This different set of simple roots $\bar{a}_i$ yields a Cartan matrix that is the   Galois conjugate of the Cartan matrix of the simple roots ${a}_i$, and leads to the label of $\frac{5}{2}$ in the Coxeter diagram.
However,  in this finite-dimensional case the groups generated by the two sets of generators are  isomorphic, as both sets of roots give rise to the same root system, but positioned differently in space.  
The non-trivial Coxeter relations  are therefore in both cases $(s_is_j)^{m_{ij}}=1$ with $m_{ij}=5$ for $A_{ij}=-\tau, \sigma$. Hence the Coxeter matrix is identical in both cases; in particular, it is still symmetric. 
Therefore, one usually restricts analysis to the standard-representation, where the simple roots form obtuse angles, and whose Cartan matrix therefore satisfies the usual negativity requirements. 
The projection $\pi_\perp$ into the second $H_4$-invariant subspace in terms of the bases $\alpha_i$ and $\bar{a}_i$ is  the Galois conjugate of $\pi_\parallel$ in Eq. (\ref{projpi}), which, however, is with respect to the bases $\alpha_i$ and ${a}_i$.
The relationship between $E_8$ and $H_4$ is best understood as the standard representation of $E_8$ inducing both the standard and non-standard representations of the subgroup $H_4$. That is, $E_8$ decomposes under an $H_4$ subgroup as $\mathbf{8}=\mathbf{4}+\bar{\mathbf{4}}$ (c.f. also \cite{Koca:1998, Koca:2001}, in particular in the wider context of similar constructions unified in  the Freudenthal-Tits magic square \cite{Baez2001Octonions}). Equivalent statements $\mathbf{6}=\mathbf{3}+\bar{\mathbf{3}}$ and $\mathbf{4}=\mathbf{2}+\bar{\mathbf{2}}$ hold true for the lower-dimensional cases, and have found applications in the quasicrystal literature \cite{Katz:1989, Senechal:1996}. 
In the quasicrystal setting, one usually only considers the projection $\pi_\parallel$;  in our setting one can consider projection into either invariant subspace, using $\pi_\parallel$ as well as $\pi_\perp$.

\subsection{Affine extensions of crystallographic Coxeter groups}\label{sec_std_single}

For a crystallographic Coxeter group, an affine Coxeter group can be constructed by  defining  {affine hyperplanes} $H_{\alpha_0 ,i}$ as solutions to the equations $(x \vert \alpha_0 ) =i\,$, where $x\in \mathcal{E}$, $\alpha_0\in \Phi$ and $i\in\mathbb{Z}$  \cite{McCammond2010Coxeter}.
The  nontrivial isometry of $\mathcal{E}$ that fixes $H_{\alpha_0 ,i}$ pointwise is  unique and called an affine reflection $s^{aff}_{\alpha_0 ,i}$.

\begin{defn}[Affine Coxeter group] \label{defaffcox}
An \emph{affine Coxeter group} is the extension of a Coxeter group by an \emph{affine reflection} $s^{aff}_{\alpha_0}$ whose geometric action is given by
\begin{equation}\label{reflectaff}
s^{aff}_{\alpha_0} v = \alpha_0 + v - \frac{2(\alpha_0\vert v)}{(\alpha_0\vert\alpha_0)}\alpha_0,
\end{equation}
and is generated by the extended set of generators including the new affine reflection associated with the \emph{affine root} $\alpha_0$.
This operation is not distance-preserving, and hence the group is no longer compact.
The \emph{affine Cartan matrix} of the affine Coxeter group is the Cartan matrix associated with the extended set of roots. The non-distance preserving nature of the affine reflection entails that the affine Cartan matrix is degenerate (positive semi-definite), and thus fulfils $\det A =0$. If the group contains both $s^{aff}_{\alpha_0}$ and $s_{\alpha_0}$, it also includes the translation generator $Tv=v+\alpha_0=s^{aff}_{\alpha_0}s_{\alpha_0}v$; otherwise, $s^{aff}_{\alpha_0}s^{aff}_{-\alpha_0}$ acts as a translation of twice the length. 
\end{defn}

It is in fact possible to construct the affine Coxeter group directly from an extension of the Cartan matrix.
\begin{defn}[Kac-Moody-type affine extension]  \label{defKMtype} A \emph{Kac-Moody-type affine extension $A^{aff}$ of a Cartan matrix} is an extension of the Cartan matrix $A$ of a Coxeter group by further rows and columns such that the following conditions hold: 
\begin{itemize}
\item  The diagonal entries are normalised  as $A^{aff}_{ii} = 2$ according to the definition  in Eq.~(\ref{CM}).
\item  The additional matrix entries of $A^{aff}$ take values in the same integer ring as the entries of $A$. This includes potentially integer rings of extended number fields as in the case of $H_3$.
\item  For off-diagonal entries we have $A^{aff}_{ij}\le 0$; moreover, $A^{aff}_{ij} = 0\Leftrightarrow A^{aff}_{ji} = 0$.
\item  The affine extended matrix fulfils the \emph{determinant constraint} $\det A^{aff}=0$.
\end{itemize}
\end{defn}

In our previous paper \cite{DechantTwarockBoehm2011H3aff} we have laid out a rationale for Kac-Moody-type extensions of Cartan matrices, as well as consistency conditions that lead to a somewhat improved algorithm for numerically searching for such matrices. This was necessitated by our search for novel asymmetric affine extensions of  $H_2$, $H_3$ and $H_4$. Here, our algorithm simply recovers the affine extensions of $E_8$, $D_6$ and $A_4$ that are  well known in the literature for  affine extensions by a single node. However,  based on  Definition \ref{defKMtype}, we will also consider extending by two nodes in the context of the projection.

We begin with the case of $E_8$, which is the most interesting from a high energy physics point of view, and the largest exceptional Coxeter group. Various notations are used in the literature to denote its unique (standard) affine extension, but here we shall use $E_8^=$, where the equality sign is meant to signify that the extra root has the same length as the other roots, i.e. the affine extension is simply-laced (see Fig. \ref{figE8aff} for our notation). The affine root $\alpha_0$ that gives rise to this affine extension can  be expressed in terms of the root vectors of $E_8$ as  
\begin{equation}
-\alpha_0=2\alpha_1+3\alpha_2+4\alpha_3+5\alpha_4+6\alpha_5+4\alpha_6+2\alpha_7+3\alpha_8,\label{affroote8}
\end{equation}
which will prove important in the projection context later. 

\condcomment{\boolean{includefigs}}{
\condcomment{\boolean{tikzfigs}}{
\begin{figure}
	\begin{center}
\begin{tikzpicture}[scale=0.5,
knoten/.style={        circle,      inner sep=.1cm,        draw}
]
\node at (-1,.7) (knoten0) [knoten,  color=white!0!black] {};
\node at  (1,.7) (knoten1) [knoten,  color=white!0!black] {};
\node at  (3,.7) (knoten2) [knoten,  color=white!0!black] {};
\node at  (5,.7) (knoten3) [knoten,  color=white!0!black] {};
\node at  (7,.7) (knoten4) [knoten,  color=white!0!black] {};

\node at  (9,.7) (knoten5) [knoten,  color=white!0!black] {};
\node at (11,.7) (knoten6) [knoten,  color=white!0!black] {};
\node at (13,.7) (knoten7) [knoten,  color=white!0!black] {};
\node at (9,2.2) (knoten8) [knoten,  color=white!0!black] {};

\node at  (-1,0) (alpha0) {$\alpha_0$};
\node at  (1,0)  (alpha1) {$\alpha_1$};
\node at  (3,0)  (alpha2) {$\alpha_2$};
\node at  (5,0)  (alpha3) {$\alpha_3$};
\node at  (7,0)  (alpha4) {$\alpha_4$};
\node at  (9,0)  (alpha7) {$\alpha_5$};
\node at (11,0)  (alpha6) {$\alpha_6$};
\node at (13,0)  (alpha5) {$\alpha_7$};
\node at (9,2.8) (alpha8) {$\alpha_8$};

\path  (knoten0) edge (knoten1);
\path  (knoten1) edge (knoten2);
\path  (knoten2) edge (knoten3);
\path  (knoten3) edge (knoten4);
\path  (knoten4) edge (knoten5);
\path  (knoten5) edge (knoten6);
\path  (knoten6) edge (knoten7);
\path  (knoten5) edge (knoten8);

\end{tikzpicture} 
\end{center}
\caption[$E_8^=$]{Dynkin diagram for the standard affine extension of $E_8$, here denoted $E_8^=$.}
\label{figE8aff}
\end{figure}}

\condcomment{\boolean{psfigs}}{
\begin{figure}
	\begin{center}
		\includegraphics[width=9cm]{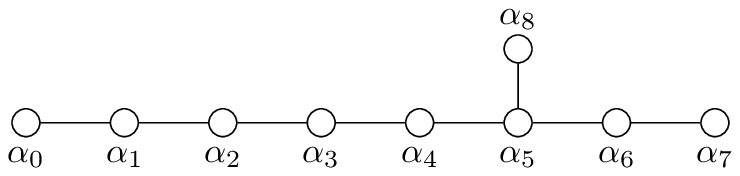}
\end{center}
\caption[$E_8^=$]{Dynkin diagram for the standard affine extension of $E_8$, here denoted $E_8^=$.}
\label{figE8aff}
\end{figure}}
}

Likewise, $D_6$ has a simply-laced affine extension, here denoted $D_6^=$ and depicted in Fig. \ref{figD6aff}. Again, the affine root can be expressed in terms of the other roots as
\begin{equation}
-\alpha_0=\alpha_1+2\alpha_2+2\alpha_3+2\alpha_4+\alpha_5+\alpha_6.\label{affroot2}
\end{equation}\

\condcomment{\boolean{includefigs}}{
\condcomment{\boolean{tikzfigs}}{
	\begin{figure}
		\begin{center}
\begin{tikzpicture}[scale=0.5,
    knoten/.style={        circle,      inner sep=.1cm,        draw}
   ]
  \node at (3,1) (knoten0) [knoten,  color=white!0!black] {};  
  \node at  (1,-.5) (knoten1) [knoten,  color=white!0!black] {};
  \node at  (3,-.5) (knoten2) [knoten,  color=white!0!black] {};
  \node at  (5,-.5) (knoten3) [knoten,  color=white!0!black] {};
  \node at  (7,-.5) (knoten4) [knoten,  color=white!0!black] {};
  \node at  (9,-.5) (knoten5) [knoten,  color=white!0!black] {};
  \node at (7,1) (knoten6) [knoten,  color=white!0!black] {};

\node at (3,1.7)  (alpha0) {$\alpha_0$};
\node at  (1,-1.2)  (alpha1) {$\alpha_1$};
\node at  (3,-1.2)  (alpha2) {$\alpha_2$};
\node at  (5,-1.2)  (alpha3) {$\alpha_3$};
\node at  (7,-1.2)  (alpha4) {$\alpha_4$};
\node at  (9,-1.2)  (alpha5) {$\alpha_5$};
\node at  (7,1.7)  (alpha6) {$\alpha_6$};

  \path  (knoten0) edge (knoten2);
  \path  (knoten1) edge (knoten2);
  \path  (knoten2) edge (knoten3);
  \path  (knoten3) edge (knoten4);
  \path  (knoten4) edge (knoten5);
  \path  (knoten4) edge (knoten6);

\node at (17,-0.5) (CM) 
{\small{$A \left(D_6^=\right)= \begin{pmatrix}
   2&0&-1&0&0&0&0  
\\ 0&2&-1&0&0&0&0  
\\ -1&-1&2&-1&0&0&0 
\\ 0&0&-1&2&-1&0&0  
\\ 0&0&0&-1&2&-1 &-1 
\\ 0&0&0&0&-1&2 &0 
\\ 0&0&0&0&-1&0&2 \end{pmatrix}$}};

\end{tikzpicture} 

\caption[$D_6^=$]{Dynkin diagram and Cartan matrix for the simply-laced standard affine extension of $D_6$ (here denoted $D_6^=$).}
\label{figD6aff}
\end{center}\end{figure}}

\condcomment{\boolean{psfigs}}{
\begin{figure}
		\includegraphics[width=14cm]{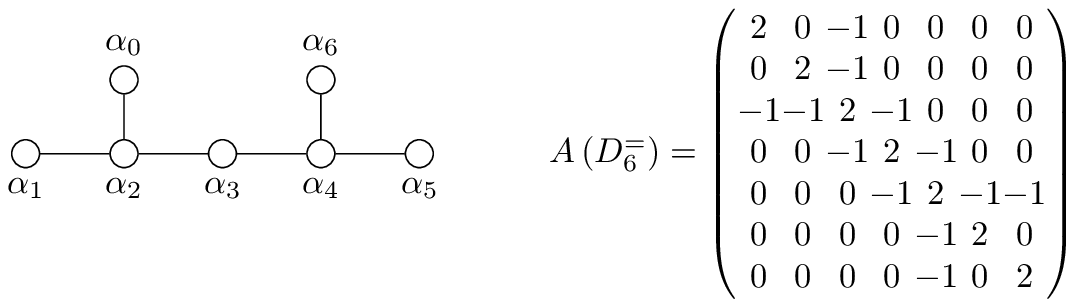}
\caption[$D_6^=$]{Dynkin diagram and Cartan matrix for the simply-laced standard affine extension of $D_6$ (here denoted $D_6^=$).}\label{figD6aff}
\end{figure}}
}

However, $D_6$ is unusual in that it also has two affine extensions with a different root length, one which we shall denote by 
$D_6^<$, because the new root is shorter than the others. In this case, the affine root is given by
\begin{equation}
-\alpha_0=\alpha_1+\alpha_2+\alpha_3+\alpha_4+\frac{1}{2}\alpha_5+\frac{1}{2}\alpha_6, \label{affroot0}
\end{equation}
There is, moreover, one with a longer root, which we denote by $D_6^>$ (both are shown in  Fig. \ref{figD6L}). Its   affine root is similarly expressible in terms of the other roots as
\begin{equation}
-\alpha_0=2\alpha_1+2\alpha_2+2\alpha_3+2\alpha_4+\alpha_5+\alpha_6.\label{affroot1}
\end{equation}

\condcomment{\boolean{includefigs}}{
\condcomment{\boolean{tikzfigs}}{
\begin{figure}
\begin{center}
	\begin{tikzpicture}[scale=0.5,
		    knoten/.style={        circle,      inner sep=.1cm,        draw}
		   ]
		  \node at (-1,.5) (knoten0) [knoten,  color=white!0!black] {};  
		  \node at  (1,.5) (knoten1) [knoten,  color=white!0!black] {};
		  \node at  (3,.5) (knoten2) [knoten,  color=white!0!black] {};
		  \node at  (5,.5) (knoten3) [knoten,  color=white!0!black] {};
		  \node at  (7,.5) (knoten4) [knoten,  color=white!0!black] {};
		  \node at  (9,.5) (knoten5) [knoten,  color=white!0!black] {};
		  \node at (7,2) (knoten6) [knoten,  color=white!0!black] {};

		\node at (-1,-0.2)  (alpha0) {$\alpha_0$};
		\node at  (1,-0.2)  (alpha1) {$\alpha_1$};
		\node at  (3,-0.2)  (alpha2) {$\alpha_2$};
		\node at  (5,-0.2)  (alpha3) {$\alpha_3$};
		\node at  (7,-0.2)  (alpha4) {$\alpha_4$};
		\node at  (9,-0.2)  (alpha5) {$\alpha_5$};
		\node at  (7,2.7)  (alpha6) {$\alpha_6$};
		\node at (0,0.35)  (midpoint) {};
		\node at (0.5,1.1)  (up) {};
		\node at (0.5,-0.1)  (down) {};

			\path (knoten0.north east) edge (knoten1.north west);
			\path (knoten0.south east) edge (knoten1.south west);
		  \path  (knoten1) edge (knoten2);
		  \path  (knoten2) edge (knoten3);
		  \path  (knoten3) edge (knoten4);
		  \path  (knoten4) edge (knoten5);
		  \path  (knoten4) edge (knoten6);
		  \path  (midpoint.mid) edge (up);
		  \path  (midpoint.mid) edge (down);

	\node at (17,0.5) (CM) 
	{\small{$A \left(D_6^<\right) = \begin{pmatrix}
	   2&-2&0&0&0&0&0  
	\\ -1&2&-1&0&0&0&0  
	\\ 0&-1&2&-1&0&0&0 
	\\ 0&0&-1&2&-1&0&0  
	\\ 0&0&0&-1&2&-1 &-1 
	\\ 0&0&0&0&-1&2 &0 
	\\ 0&0&0&0&-1&0&2 \end{pmatrix}$}};

	\end{tikzpicture} \\

	\begin{tikzpicture}[scale=0.5,
	    knoten/.style={        circle,      inner sep=.1cm,        draw}
	   ]
		  \node at (-1,.5) (knoten0) [knoten,  color=white!0!black] {};  
		  \node at  (1,.5) (knoten1) [knoten,  color=white!0!black] {};
		  \node at  (3,.5) (knoten2) [knoten,  color=white!0!black] {};
		  \node at  (5,.5) (knoten3) [knoten,  color=white!0!black] {};
		  \node at  (7,.5) (knoten4) [knoten,  color=white!0!black] {};
		  \node at  (9,.5) (knoten5) [knoten,  color=white!0!black] {};
		  \node at (7,2) (knoten6) [knoten,  color=white!0!black] {};

		\node at (-1,-0.2)  (alpha0) {$\alpha_0$};
		\node at  (1,-0.2)  (alpha1) {$\alpha_1$};
		\node at  (3,-0.2)  (alpha2) {$\alpha_2$};
		\node at  (5,-0.2)  (alpha3) {$\alpha_3$};
		\node at  (7,-0.2)  (alpha4) {$\alpha_4$};
		\node at  (9,-0.2)  (alpha5) {$\alpha_5$};
		\node at  (7,2.7)   (alpha6) {$\alpha_6$};
				\node at (0,0.35)  (midpoint) {};
				\node at (-0.5,1.1)  (up) {};
				\node at (-0.5,-0.1)  (down) {};
		\path (knoten0.north east) edge (knoten1.north west);
		\path (knoten0.south east) edge (knoten1.south west);
		  \path  (knoten1) edge (knoten2);
		  \path  (knoten2) edge (knoten3);
		  \path  (knoten3) edge (knoten4);
		  \path  (knoten4) edge (knoten5);
		  \path  (knoten4) edge (knoten6);
				  \path  (midpoint.mid) edge (up);
				  \path  (midpoint.mid) edge (down);

	\node at (17,0.5) (CM) 
	{\small{$A \left(D_6^>\right) = \begin{pmatrix}
	   2&-1&0&0&0&0&0  
	\\ -2&2&-1&0&0&0&0  
	\\ 0&-1&2&-1&0&0&0 
	\\ 0&0&-1&2&-1&0&0  
	\\ 0&0&0&-1&2&-1 &-1 
	\\ 0&0&0&0&-1&2 &0 
	\\ 0&0&0&0&-1&0&2 \end{pmatrix}$}};

	\end{tikzpicture}

\caption[$D_6^<$]{Dynkin diagrams and Cartan matrices for the standard affine extensions of $D_6$ with a short affine root (here denoted $D_6^<$), and that with a long affine root, here denoted $D_6^>$. The arrow conventionally points to the shorter root.} 
\label{figD6L}
\end{center}
\end{figure}}
\condcomment{\boolean{psfigs}}{
\begin{figure}
	\begin{center}
	\includegraphics[width=16cm]{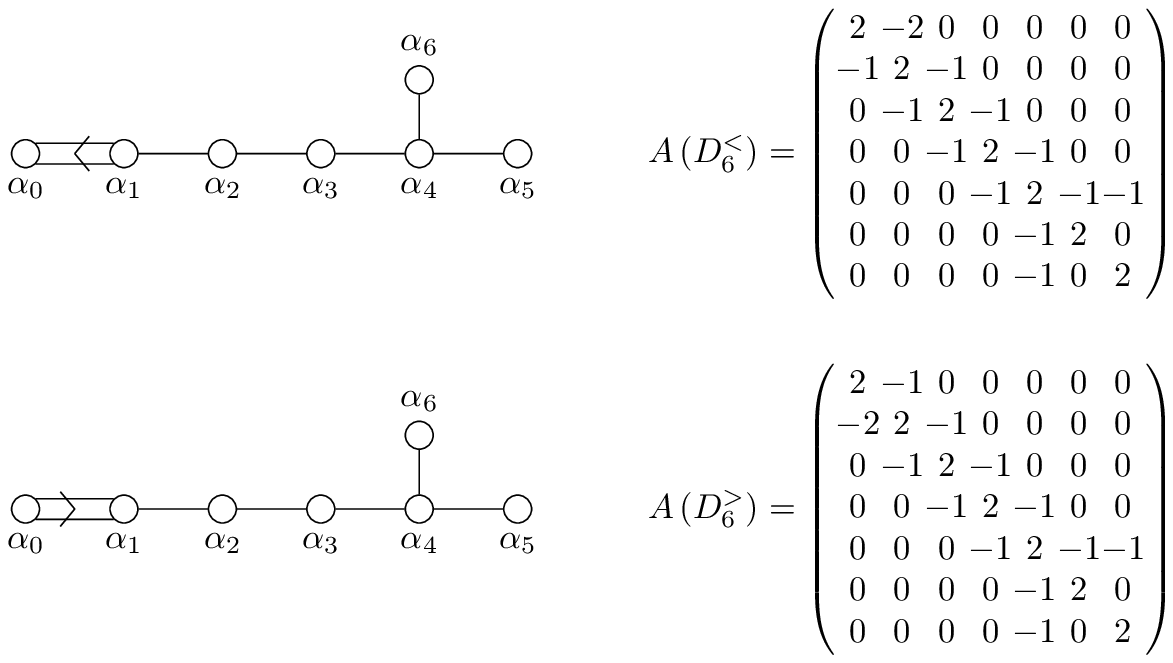}

\caption[$D_6^<$]{Dynkin diagrams and Cartan matrices for the standard affine extensions of $D_6$ with a short affine root (here denoted $D_6^<$), and that with a long affine root, here denoted $D_6^>$. The arrow conventionally points to the shorter root.} 
\label{figD6L}
\end{center}
\end{figure}}
}

\condcomment{\boolean{includefigs}}{
\condcomment{\boolean{tikzfigs}}{
\begin{figure}
\begin{center}
	\begin{tikzpicture}[scale=0.5,
		    knoten/.style={        circle,      inner sep=.1cm,        draw}
		   ]
		  \node at (4,2.0) (knoten0) [knoten,  color=white!0!black] {};
		  \node at (1,0.5) (knoten1) [knoten,  color=white!0!black] {};
		  \node at (3,0.5) (knoten2) [knoten,  color=white!0!black] {};
		  \node at (5,0.5) (knoten3) [knoten,  color=white!0!black] {};
		  \node at (7,0.5) (knoten4) [knoten,  color=white!0!black] {};

		\node at (4,2.7)  (alpha0) {$\alpha_0$};
		\node at (1,-0.2)  (alpha1) {$\alpha_1$};
		\node at (3,-0.2)  (alpha2) {$\alpha_2$};
		\node at (5,-0.2)  (alpha3) {$\alpha_3$};
		\node at (7,-0.2)  (alpha4) {$\alpha_4$};

		  \path  (knoten0) edge (knoten1);
		  \path  (knoten0) edge (knoten4);
		  \path  (knoten1) edge (knoten2);
		  \path  (knoten2) edge (knoten3);
		  \path  (knoten3) edge (knoten4);

	\node at (17,0.5) (CM) 
	{\small{$A \left(A_4^=\right) = \begin{pmatrix}
	    2&-1&0&0&-1
	\\ -1&2&-1&0&0  
	\\  0&-1&2&-1&0 
	\\  0&0&-1&2&-1
	\\ -1&0&0&-1&2
	 \end{pmatrix}$}};

	\end{tikzpicture}

\caption[$A_4^=$]{Dynkin diagram and Cartan matrix for the simply-laced standard affine extension of $A_4$, here denoted $A_4^=$.}
\label{figA4aff}
\end{center}
\end{figure}}

\condcomment{\boolean{psfigs}}{
\begin{figure}
\begin{center}

			\includegraphics[width=16cm]{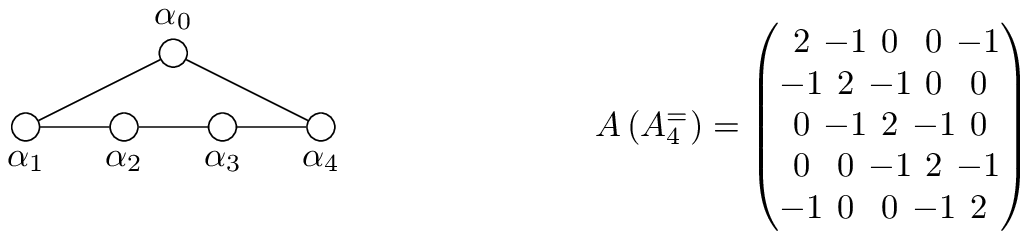}

\caption[$A_4^=$]{Dynkin diagram and Cartan matrix for the simply-laced standard affine extension of $A_4$, here denoted $A_4^=$.}
\label{figA4aff}
\end{center}
\end{figure}}
}

$A_4$ also has a unique standard affine extension, which is simply-laced and hence  will be denoted by
$A_4^=$ (see Fig. \ref{figA4aff}).
The affine root is given by
\begin{equation}
-\alpha_0=\alpha_1+\alpha_2+\alpha_3+\alpha_4.\label{affroota4}
\end{equation}

\section{Affine extensions of non-crystallographic root systems induced by projection}\label{sec_single}

In this section, we present a novel construction of affine extensions of non-crystallographic Coxeter groups, as illustrated in Fig. \ref{EDAH_overview} and indicated by the dashed arrow there. 
We induce affine extensions in the lower-dimensional, non-crystallographic case by applying the projection formalism from Section \ref{sec_proj} to the  five affine extensions from Figs \ref{figE8aff}-\ref{figA4aff} in Section \ref{sec_std_single}. We show that the induced extensions are invariant under the Dynkin diagram automorphisms of the crystallographic groups and their affine extensions  (Section \ref{sec_autos}), and that in a wider class of further extensions (simply-laced, double extensions with non-trivial projection kernel), none are compatible with the projection formalism (Section \ref{sec_further}). 

\subsection{Projecting the affine root}\label{sec_proj_aff}

In the previous section, we have introduced the projection formalism, and we have presented the standard affine extensions of the relevant crystallographic groups. In particular, in each case we have given expressions for the affine roots in  terms of the root vectors of the unextended group. By the linearity of the projection, one can compute the projection of the affine root. In analogy to the fact that the other roots project  to generators of the groups $H_i$ ($i=2,3,4$), we treat the projected affine root as an additional, affine,  root for the projected group $H_i$, thereby inducing an affine extension of $H_i$.

\begin{defn}[Induced affine root] For a pair of Coxeter groups $(G^U, G^D)$ related via projection, i.e. a non-degenerate mapping $\pi$ of the root system of  $G^U$ onto the root system of $G^D$, we call the projection of the affine root of an affine extension of $G^U$ the \emph{induced (affine) root} of $G^D$. The matrix defined as in Eq. (\ref{CM}) by the induced affine root and the simple roots of $G^D$ define the \emph{induced (affine) Cartan matrix}.
\end{defn}

\begin{thm}[Induced Extensions] The five affine extensions $A_4^=$, $D_6^>$, $D_6^=$, $D_6^<$ and $E_8^=$ of  $A_4$, $D_6$ and $E_8$  induce affine extensions of $H_2$, $H_3$ and $H_4$ via the projections linking the respective root systems. For $\pi_\parallel$, these five induced extensions shall be denoted by $H_2^=$, $H_3^>$, $H_3^=$, $H_3^<$ and $H_4^=$. Projection with $\pi_\perp$ into the other invariant subspace yields affine roots that are the Galois conjugates of those induced by $\pi_\parallel$, and the five corresponding induced affine extensions shall be denoted by $\bar{H}_2^=$, $\bar{H}_3^>$, $\bar{H}_3^=$, $\bar{H}_3^<$ and $\bar{H}_4^=$.
\end{thm}

\begin{proof} We consider the five cases in turn. 
\begin{enumerate}
\item We  begin with the case of $E_8$. 
We have shown above that the root vectors can be projected onto the $H_4$ root vectors $a_i$ by the projection $\pi_\parallel$ shown in Fig. \ref{figE8}. 
The projection in Eq. (\ref{projpi}) of the affine root in Eq. (\ref{affroote8}) is therefore
\begin{equation}
-a_0=\pi_\parallel(-\alpha_0)= 2(1+\tau)a_1+(3+4\tau)a_2+2(2+3\tau)a_3+(3+5\tau)a_4.\label{E8projaff}
\end{equation}
Using the inner products from the $H_4$ Cartan matrix   $(a_1|a_2)=-\frac{1}{2}$,  $(a_2|a_3)=-\frac{1}{2}$ and $(a_3|a_4)=-\frac{\tau}{2}$, the inner products of the additional root with the 
roots of $H_4$ are $(a_0|a_1)=-\frac{1}{2}$ and  $(a_0|a_2)= (a_0|a_3)= (a_0|a_4)=0$. 
 Thus, the  Cartan matrix corresponding to the simple roots of  $H_4$ extended by the projected  affine root of $E_8^=$ is found to be
\begin{equation}
A\left(H_4^=\right) := \begin{pmatrix} 2&\tau-2&0&0&0 \\ -1&2&-1&0&0 \\ 0&-1&2&-1&0 \\ 0&0&-1&2&-\tau \\ 0&0&0&-\tau&2 \end{pmatrix}.\label{CarH4affAS}
\end{equation}
This is  one of the Kac-Moody-type affine extensions of $H_4$  that we derived in our previous paper \cite{DechantTwarockBoehm2011H3aff},  in the context of non-crystallographic Coxeter groups. It was listed there as the first non-trivial example of affine extensions of this type and corresponds to an affine extension  of length $\tau$ along the highest root $\alpha_H$ of $H_4$. We will briefly review the results from  \cite{DechantTwarockBoehm2011H3aff} in Section \ref{sec_sum}, which we will use to classify all induced affine extensions in Section \ref{sec_class}. 

Projecting with $\pi_\perp$ into the other $H_4$-invariant subspace spanned by the basis of simple roots $\bar{a}_i$  yields the Galois conjugate of the affine root in Eq. (\ref{E8projaff})
\begin{equation}
-\bar{a}_0=\pi_\perp(-\alpha_0)= 2(1+\sigma)\bar{a}_1+(3+4\sigma)\bar{a}_2+2(2+3\sigma)\bar{a}_3+(3+5\sigma)\bar{a}_4.\label{E8projaff2}
\end{equation}
Using the inner products   $(\bar{a}_1|\bar{a}_2)=-\frac{1}{2}$,  $(\bar{a}_2|\bar{a}_3)=-\frac{1}{2}$ and $(\bar{a}_3|\bar{a}_4)=-\frac{\sigma}{2}$, the inner products of the affine root with the 
 $H_4$ roots are $(\bar{a}_0|\bar{a}_1)=-\frac{1}{2}$ and  $(\bar{a}_0|\bar{a}_2)= (\bar{a}_0|\bar{a}_3)= (\bar{a}_0|\bar{a}_4)=0$, and the resulting Cartan matrix $A\left(\bar{H}_4^=\right)$  is thus the Galois conjugate of that in Eq. (\ref{CarH4affAS}) from the other invariant subspace. Since both sets of roots $a_i$ and $\bar{a}_i$ generate the same abstract group $H_4$, one  has a pair of Galois conjugate induced affine roots $\alpha_0$ and $\bar{\alpha}_0$ parallel to the highest root $\alpha_H$ with Galois conjugate lengths $\tau$ and $\sigma$, respectively. Note that $\left(A\left(H_4^=\right) \right)^T$ would also generate the same translation of length $\sigma$ along $\alpha_H$, and was contained in the results of  \cite{DechantTwarockBoehm2011H3aff}. We will consider whether $\left(A\left(H_4^=\right) \right)^T$ could also arise from projection in Section \ref{sec_concl}. 

\item
Using the same procedure as above -- i.e.  employing linearity to project the affine root of $D_6^=$ and using it as an affine extension of $H_3$ -- generates the  analogue of the previous case in three dimensions
\begin{equation}
A\left(H_3^=\right) :=  \begin{pmatrix} 2&0&\tau-2 &0 \\ 0&2&-1&0 \\ -1&-1&2&-\tau\\ 0&0&-\tau&2 \end{pmatrix}\label{CarH3affAS2}.
\end{equation}
This is also the first non-trivial example of asymmetric affine extensions of $H_3$ considered in our previous paper  \cite{DechantTwarockBoehm2011H3aff}, corresponding to an affine extension of length $\tau$ along the highest root $\alpha_H$ of $H_3$ (i.e. along a 2-fold axis of icosahedral symmetry). One  choice of simple roots for $H_3$ is $\alpha_1=(0,1,0)$, $\alpha_2=-\frac{1}{2}(-\sigma,1,\tau)$, and $\alpha_3=(0,0,1)$, for which  $\alpha_H=(1,0,0)$.  

Projection into the other invariant subspace likewise generates the Galois conjugate affine root $\bar{\alpha}_0$ and the Galois conjugate Cartan matrix  $A\left(\bar{H}_3^=\right)$,  thereby giving rise to a translation of length $\sigma$ along  $\alpha_H$.

\item
When projecting $D_6^<$ we find
\begin{equation}
A\left(H_3^<\right) := \begin{pmatrix} 2&\frac{4}{5}(\tau-3) &0 &0 \\ -1&2&-1&0 \\ 0&-1&2&-\tau\\ 0&0&-\tau&2 \end{pmatrix}.\label{CarH3affAS0}
\end{equation}
In  \cite{DechantTwarockBoehm2011H3aff}, we have considered a family of matrices of this form analytically and found a similar classification as in the other cases, according to a certain Fibonacci scaling relation  (c.f. Section \ref{sec_sum}). Note,  that the projection construction here naturally leads to $\mathbb{Q}[\tau]$-valued    entries of the Cartan matrix, suggesting to analyse this more general class of Cartan matrices over the extended number field  $\mathbb{Q}[\tau]=\lbrace a+\tau b| a,b \in \mathbb{Q}\rbrace$.  Cartan matrices of this form correspond to affine extensions along a 5-fold axis of icosahedral symmetry $T_5$, where $T_5=(\tau, -1, 0)$ in our chosen basis of simple roots, and the normalisation is chosen for later convenience. The affine root of $H_3^<$ is then given by $\alpha_0=\frac{1}{2}T_5$. 

Projection into the other invariant subspace again yields the Galois conjugate $\bar{\alpha}_0$ of $\alpha_0$,  corresponding to $\bar{\alpha}_0=-\frac{1}{2}\sigma T_5$ for our normalisation of $T_5$.

\item
A similar result is obtained when $\pi_\parallel$-projecting $D_6^>$ to $H_3^>$
\begin{equation}
A\left(H_3^>\right) :=\begin{pmatrix} 2&\frac{2}{5}(\tau-3)&0 &0 \\ -2&2&-1&0 \\ 0&-1&2&-\tau\\ 0&0&-\tau&2 \end{pmatrix}.\label{CarH3affAS1}
\end{equation}
The respective projections again yield the Galois conjugate pair $\alpha_0=T_5$ and $\bar{\alpha}_0=-\sigma T_5$.
We note that even though  the affine extensions $D_6^<$ and $D_6^>$ are related by transposition, the correspondence between the two induced affine extensions  $H_3^<$ and $H_3^>$ (and $\bar{H}_3^<$ and $\bar{H}_3^>$) is not so straightforward, i.e. the operations of transposition and projection do not commute, schematically $[T,P]\ne 0$. 
However, the transposed versions of these induced lower-dimensional Cartan matrices, for instance $A\left(H_3^>\right)^T$, were among the affine extensions derived in \cite{DechantTwarockBoehm2011H3aff}. 
One might therefore wonder  which higher-dimensional Cartan matrices could give rise to these transposed versions after projection. We will revisit these issues later.

\item
The affine root of $A_4$ is given by Eq. (\ref{affroota4})
and upon projection with $\pi_\parallel$ yields an affine extension of $H_2$ analogous to the other simply-laced cases considered above
\begin{equation}
A\left(H_2^=\right) := \begin{pmatrix} 2&\tau-2&\tau-2 \\ -1&2&-\tau\\ -1&-\tau&2 \end{pmatrix}.\label{CarH2affAS}
\end{equation}
This likewise corresponds to an affine root of length $\tau$ along $\alpha_H$, the highest root of $H_2$ given by $\alpha_H=\tau(\alpha_1+\alpha_2)$, which was also found in \cite{DechantTwarockBoehm2011H3aff}, where we also visualised its action on a pentagon. 
Projection with $\pi_\perp$ yields an induced extension $\bar{H}_2^=$ with the Galois-conjugate length $-\sigma$ along $\alpha_H$. 
\end{enumerate}

This completes the proof.
\end{proof}

As is well known \cite{Humphreys1990Coxeter}, the affine extensions of crystallographic Coxeter groups result in a (periodic) tessellation of the fundamental domain of the unextended group in terms of copies of the  fundamental domain of the affine group.
In contrast, affine extended non-crystallographic groups  inherit the full fundamental domain of the unextended group. 
The fundamental domain of these extensions however still has the interesting property of being 
tessellated, but in this case  the  tiling is \emph{aperiodic}, and hence the fundamental domain again has a non-trivial mathematical structure. 

In order to further explore this interesting relation with quasilattices, we begin by introducing some terminology \cite{Twarock:2002a}. We recall that a generic affine non-crystallographic Coxeter group $H_i^+$ is generated by the $s_j$s from Section \ref{sec_Cox} together with the translation $T$ that we identified in Definition \ref{defaffcox}. 

\begin{defn}[Quasicrystal fragment] \label{defQCF}
	Let $\Phi$ denote the root polytope of the non-crystallographic Coxeter group $H_i$, and let $W^m(s_j;T)$ denote the set of all \emph{words} $w(s_j;T)$ in the \emph{alphabet} formed from the \emph{letters} $s_j$ and $T$  in which $T$ appears precisely $m$ times. The set of points 
	\begin{equation} Q_i (n) := \lbrace W^m 	(s_j;T )\Phi 	| m \le n\rbrace \end{equation}
is called an \emph{$H_i^+$-induced quasicrystal fragment}; $n$ is the \emph{cut-off-parameter}. 
The \emph{cardinality} of such a quasicrystal fragment will be denoted by $|Q_i (n)|$, and a generic translation yields the maximal cardinality $|Q_i^{max} (n)|$. We will say that a quasicrystal fragment with less than maximal cardinality   \emph{has coinciding/degenerate points/vertices} and we call the corresponding translation \emph{distinguished}.  This degeneracy implies \emph{non-trivial relations} $w_1(s_j;T)v=w_2(s_j;T)v$ (for $v \in \Phi$) amongst the (words in the) generators. The set of points  $P_i (n) := \lbrace W^m 	(s_j;T )R 	| m = n\rbrace$ will denote the \emph{shell} of the quasicrystal fragment determined by the words that contain $T$  precisely $n$ times.
\end{defn}

\begin{table}
\begin{centering}\begin{tabular}{|c||c|c|c|}
\hline
$\lambda$&$H_2$&$H_3$&$H_4$
\tabularnewline
\hline
\hline
$0$&$10$&$30$&$120$
\tabularnewline
\hline
\hline
$\sigma$&$40$&$552$&$5280$
\tabularnewline
\hline
$1$&$36$&$361$&$3721$
\tabularnewline
\hline
$\tau$&$40$&$552$&$5280$
\tabularnewline
\hline
\end{tabular}\par\end{centering}
\caption[E8FibClass]{\label{tabFibCard} Cardinalities of extended root systems/quasicrystal fragments depending on translation length. 
Here, we list cardinalities $|P_i (1)|$ of the point sets achieved by extending the $H_i$ root systems by an affine reflection along the highest root  $\alpha_0=-\lambda \alpha_H$ for various values of $\lambda$.  $\lambda=1$ corresponds to the simply-laced affine extension  $H_i^{aff}$ that was considered in  \cite{Twarock:2002a}. The induced extensions $H_i^{=}$ and $\bar{H}_i^=$ considered here correspond to $\lambda=\tau$ and   $\lambda=-\sigma$, respectively, and yield the same cardinality.    $\left(H_i^{aff}\right)^T$ is also an affine extension corresponding to  $\lambda=-\sigma$, and is contained in the solutions found in \cite{DechantTwarockBoehm2011H3aff}.  We will  discuss how this could be lifted to the higher-dimensional case in Section \ref{sec_concl}. We note that all three translations in each case are distinguished, i.e. they give rise to less than maximal cardinality. }
\end{table}
%

The  affine roots relevant here are all parallel to  the respective highest root $\alpha_H$   but have various different lengths $\lambda$, which we write  as $\alpha_0=-\lambda \alpha_H$. 
 Therefore, in Table \ref{tabFibCard} we present the cardinalities $|P_i (1)|$ of point arrays   derived from the root systems of $H_2$, $H_3$ and $H_4$ (the decagon, the icosidodecahedron and the 600-cell, respectively) for translation lengths $\lambda=\lbrace 0,\sigma,1,\tau\rbrace$. 
$\lambda=0$ corresponds to the unextended group, and the induced affine extensions $H_i^=$ from Eqs (\ref{CarH2affAS}), (\ref{CarH3affAS2}) and (\ref{CarH4affAS}) considered here correspond to $\lambda=\tau$.
The simply-laced extensions $H_i^{aff}$ considered in  \cite{Twarock:2002a} have $\lambda=1$. 
The transposes of $A(H_i^=)$ are affine extensions with length $\lambda=-\sigma$ that were also amongst those found in \cite{DechantTwarockBoehm2011H3aff}. They are also equivalent to the induced affine extensions from the other invariant subspace, $\bar{H}^=_i$, as the compact part of the group is the same and they give rise to the same translations $\lambda=-\sigma$.  This forms a subset of the extensions found in  \cite{DechantTwarockBoehm2011H3aff} that is distinguished via the projection and also through its symmetric place in the Fibonacci classification in  \cite{DechantTwarockBoehm2011H3aff}, which we will discuss further in Sections \ref{sec_sum} and \ref{sec_class}. 
All three translations belonging to the special three cases of $H_i^{aff}$ and the induced $H_i^=$ and $\bar{H}_i^=$ are found to be \emph{distinguished}. 
We also note that the Galois-conjugate translations yield the same cardinalities, i.e. the rows corresponding to $\sigma$ and $\tau$ have identical entries. Investigating the corresponding cases for $H_3^<$, one finds that the Galois conjugate affine roots $\alpha_0=\frac{1}{2}T_5$ and  $\bar{\alpha}_0=-\frac{1}{2}\sigma T_5$ yield the same cardinality of $212$. For $H_3^>$, the conjugate pair $\alpha_0=T_5$ and  $\bar{\alpha}_0=-\sigma T_5$ has the same cardinality $330$. Note that these are all distinguished translations. The cardinalities are lower than for $H_3^{aff}$, $H_3^{=}$ and $\bar{H}_3^{=}$ where  $\alpha_H=T_2=(1,0,0)$, since there are thirty $2$-fold but only twelve $5$-fold axes of icosahedral symmetry.

\begin{figure}
	\begin{center}
	      \begin{tabular}{@{}c@{ }c@{ }c@{ }}
				\begin{tikzpicture}
				\node (img) [inner sep=0pt,above right]
				{\includegraphics[width=3.5cm]{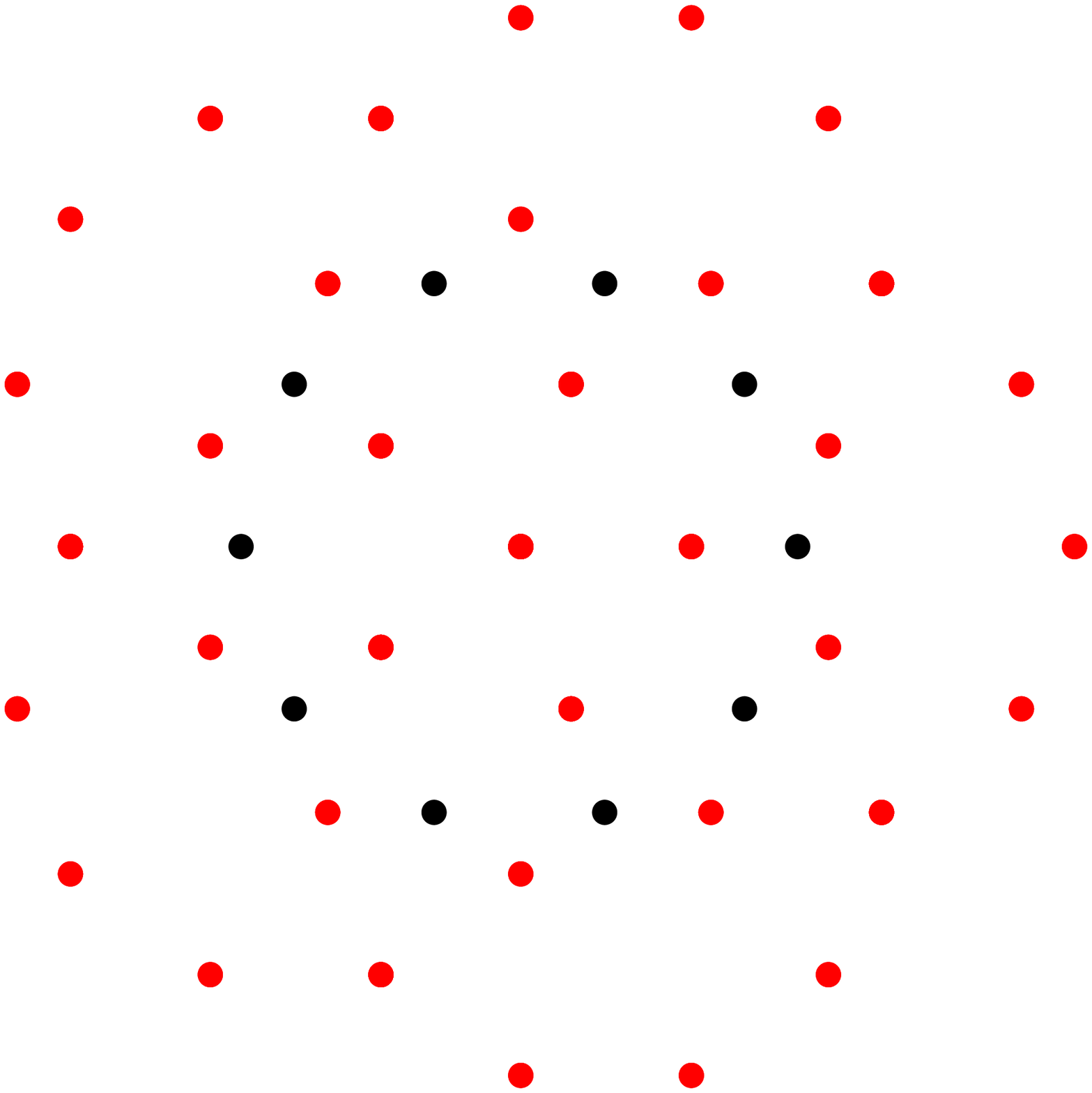}};
				\end{tikzpicture}&\hspace{1.5cm}
				\begin{tikzpicture}
				\node (img) [inner sep=0pt,above right]
				{\includegraphics[width=3.5cm]{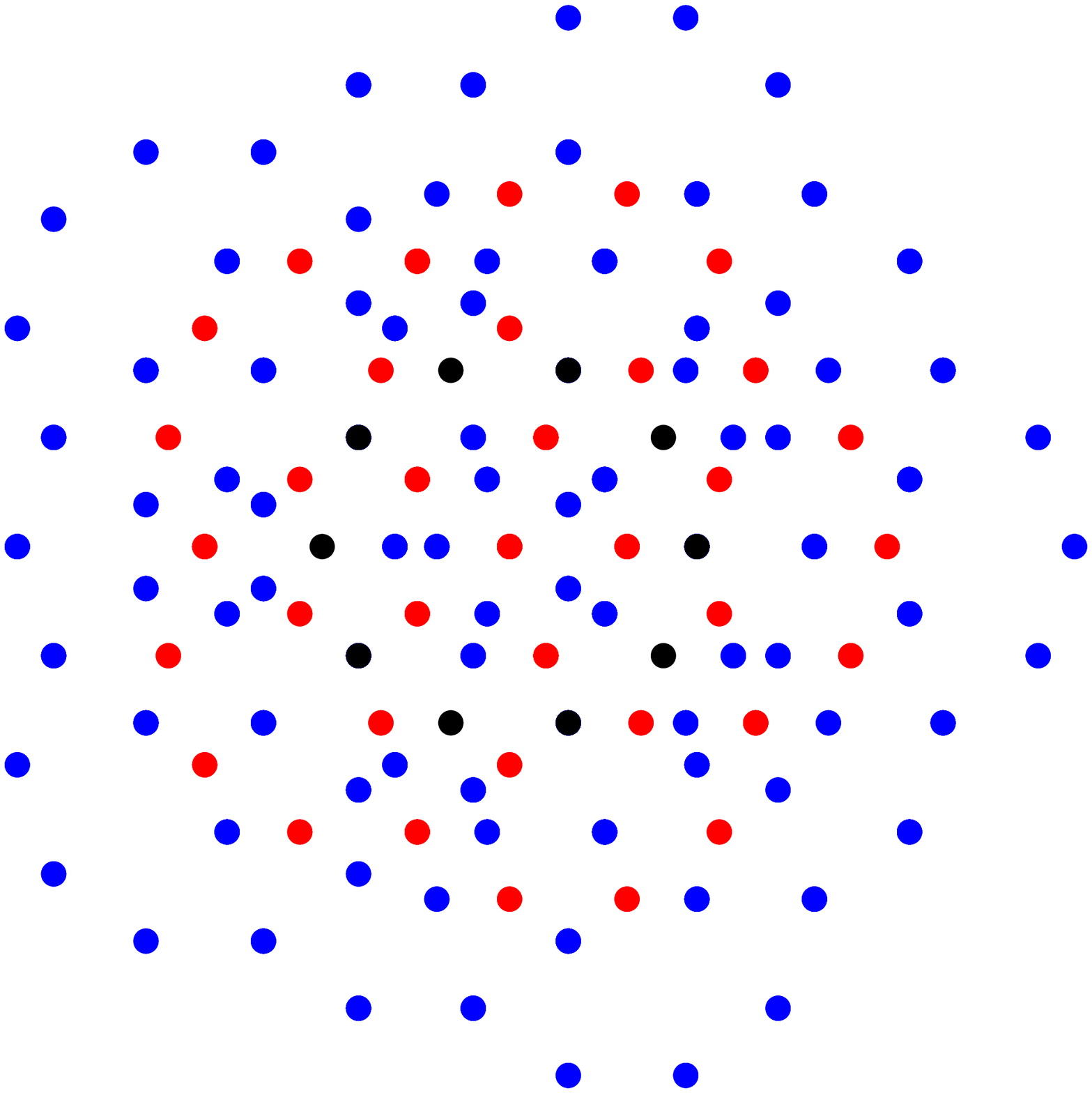}};
				\end{tikzpicture} &\hspace{1.5cm}

				\begin{tikzpicture}
				\node (img) [inner sep=0pt,above right]
				{\includegraphics[width=3.5cm]{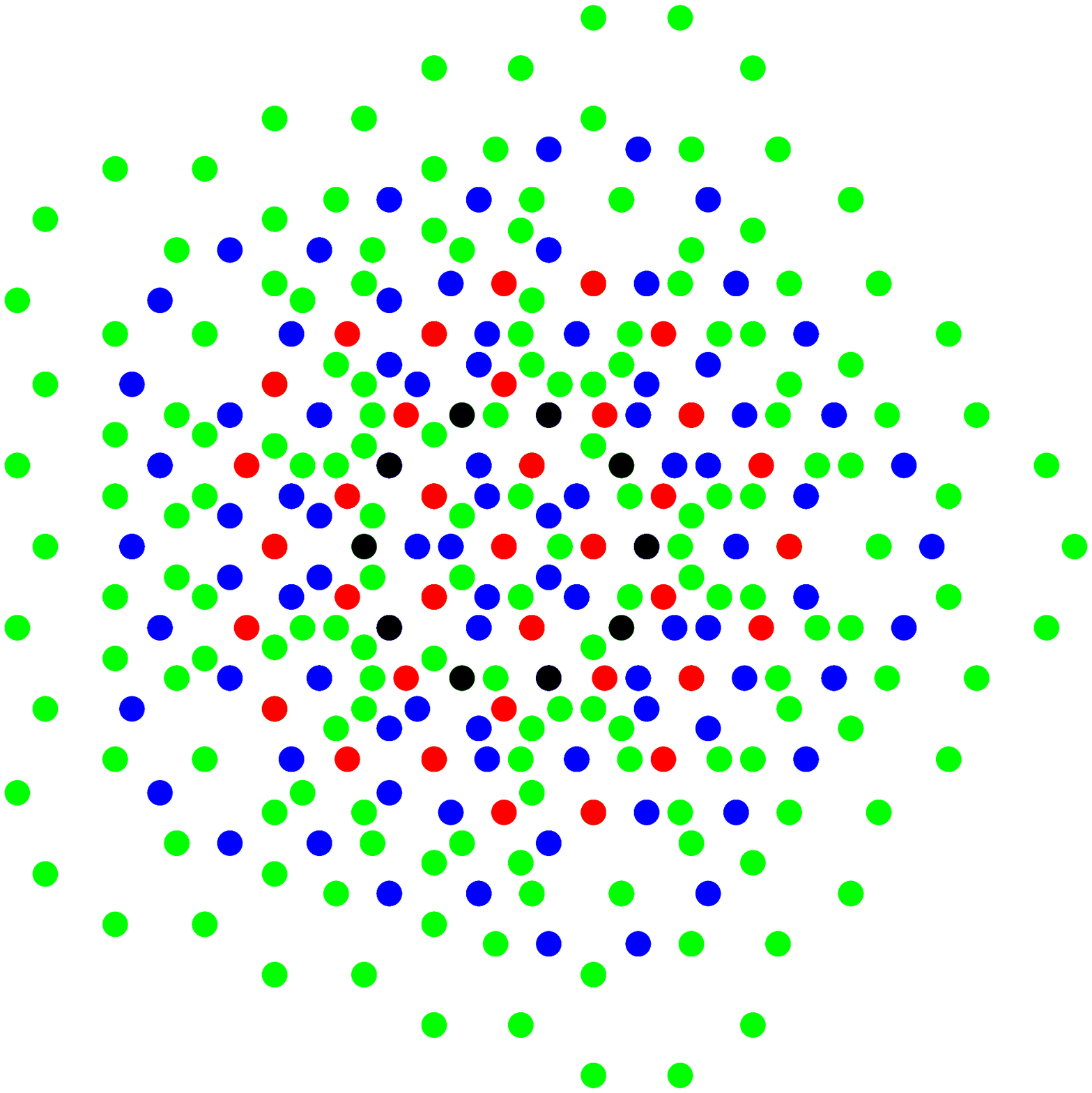}};
				\end{tikzpicture}
		\vspace{0.25cm}
			\\
			(a) $|P_2 (1)|=36$&(b) $|P_2 (2)|=90$ &(c) $|P_2 (3)|=185$\\
	 	\vspace{0.25cm}

				\begin{tikzpicture}
				\node (img) [inner sep=0pt,above right]
				{\includegraphics[width=3.5cm]{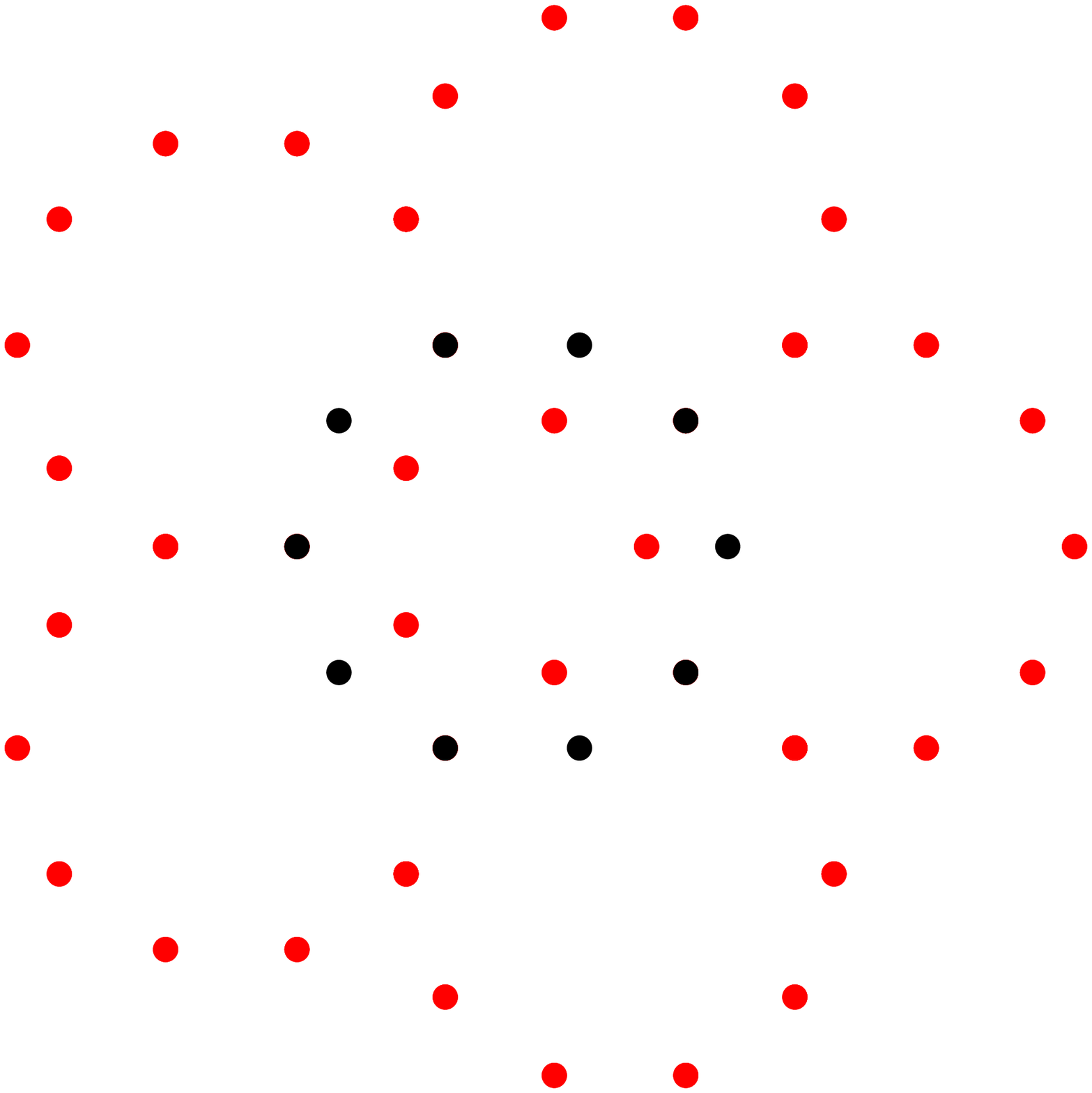}};
				\end{tikzpicture}&\hspace{1.5cm}
				\begin{tikzpicture}
				\node (img) [inner sep=0pt,above right]
				{\includegraphics[width=3.5cm]{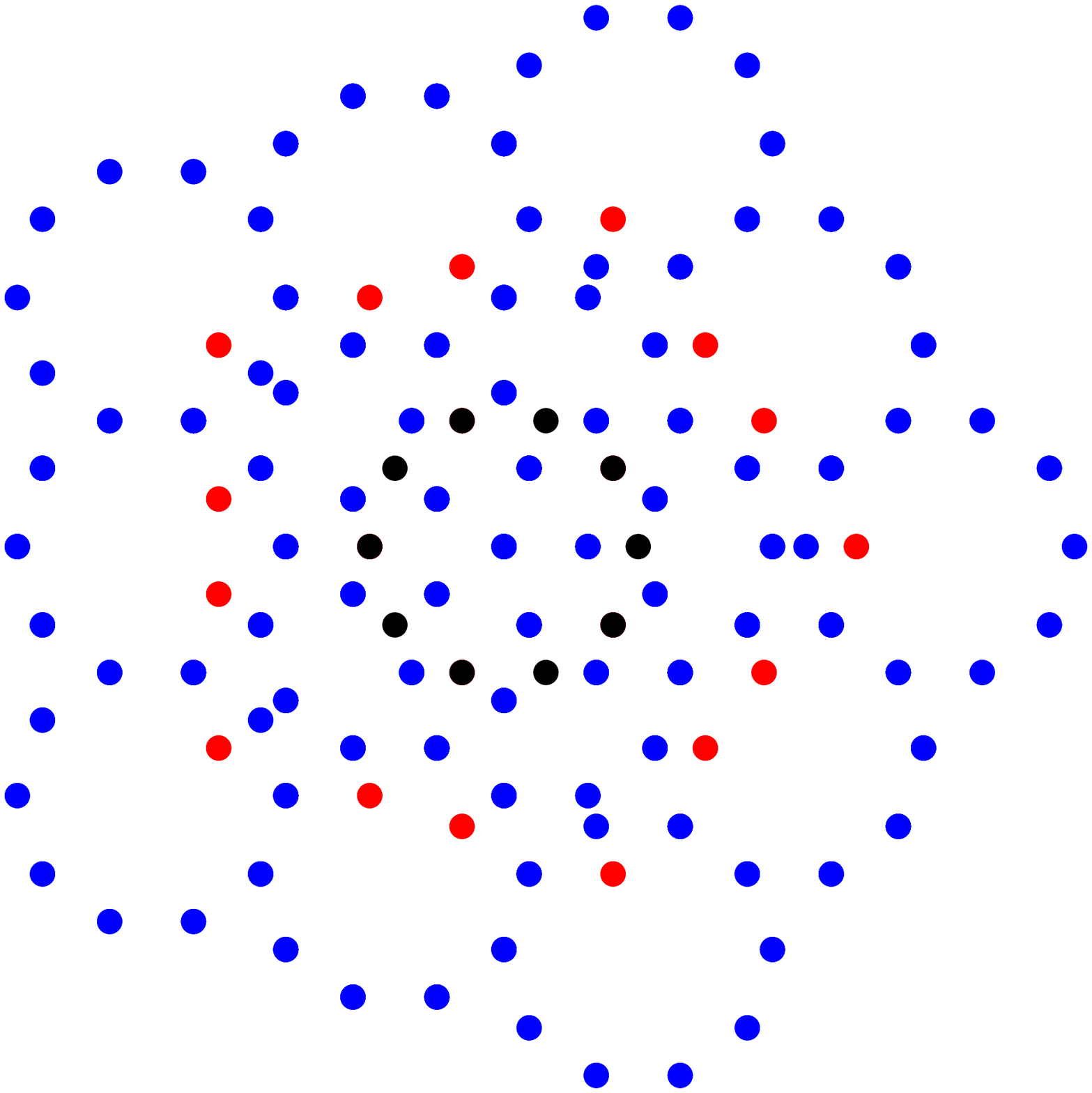}};
				\end{tikzpicture} &\hspace{1.5cm}
					\begin{tikzpicture}
					\node (img) [inner sep=0pt,above right]
					{\includegraphics[width=3.5cm]{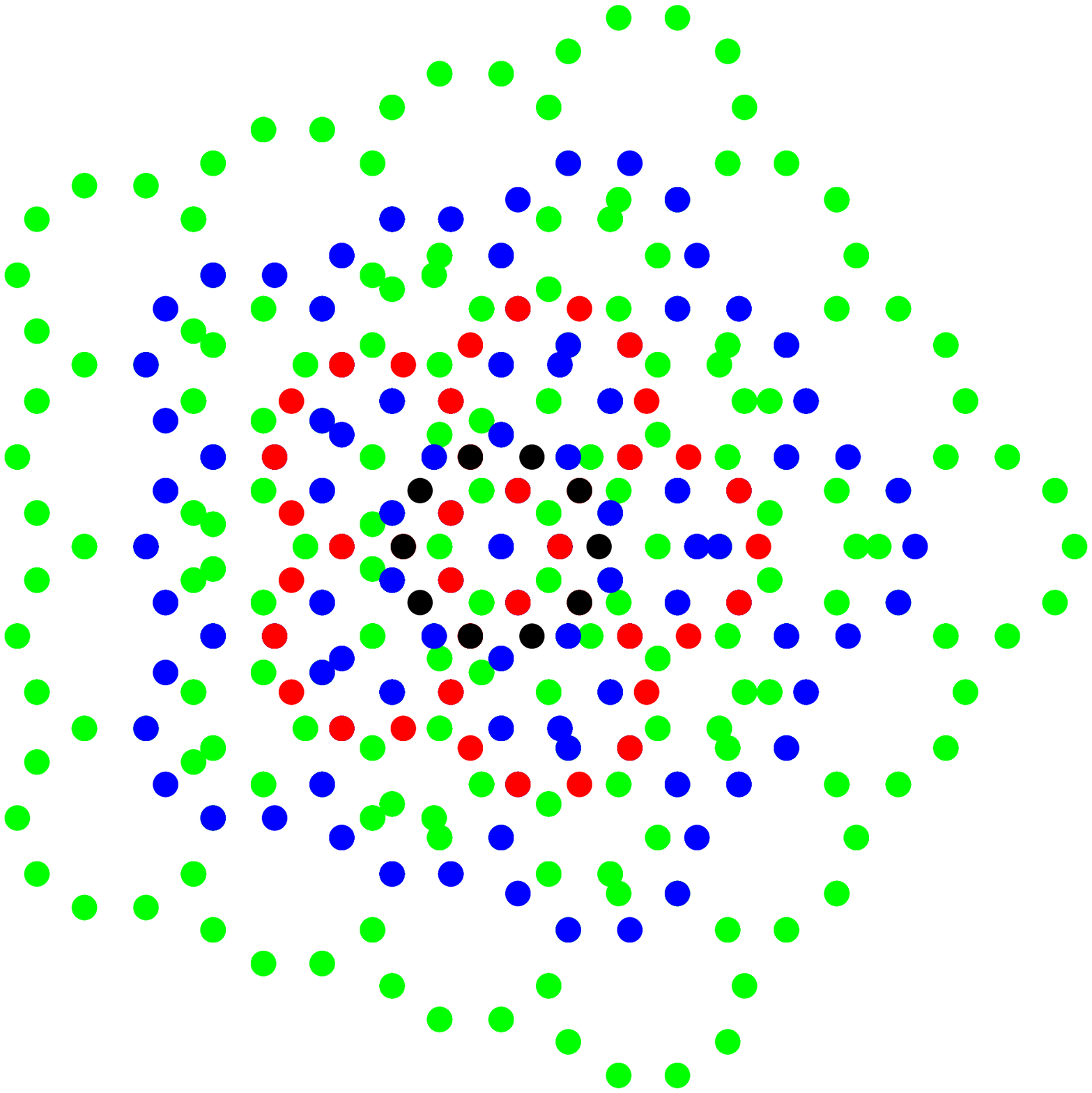}};
					\end{tikzpicture}

		\vspace{0.25cm}
			\\
			(d) $|P_2 (1)|=40$&(e) $|P_2 (2)|=101$&(f) $|P_2 (3)|=206$\\
	 	\vspace{0.25cm}

				\begin{tikzpicture}
				\node (img) [inner sep=0pt,above right]
				{\includegraphics[width=3.5cm]{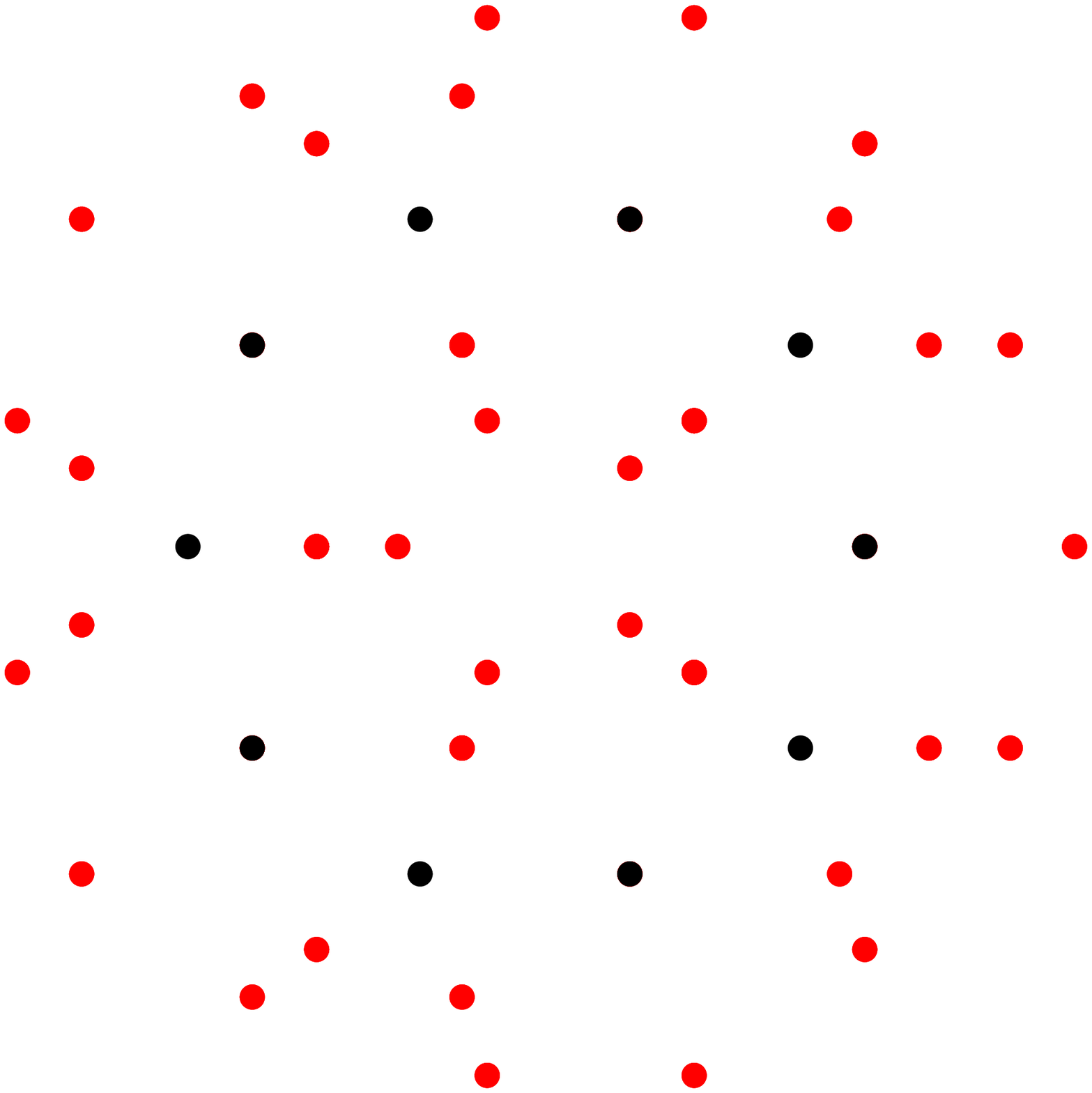}};
				\end{tikzpicture}&\hspace{1.5cm}
				\begin{tikzpicture}
				\node (img) [inner sep=0pt,above right]
				{\includegraphics[width=3.5cm]{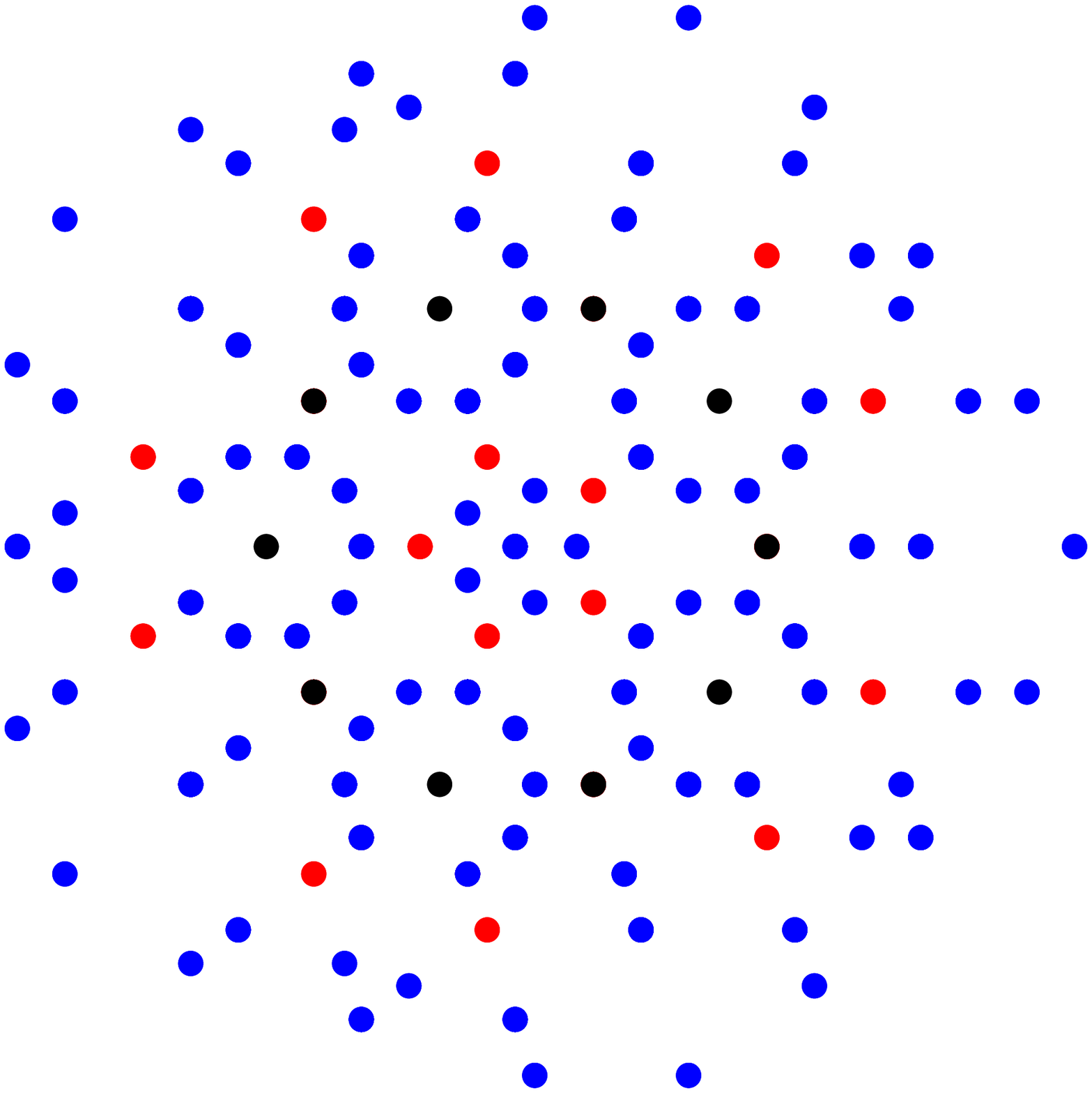}};
				\end{tikzpicture} &\hspace{1.5cm}
					\begin{tikzpicture}
					\node (img) [inner sep=0pt,above right]
					{\includegraphics[width=3.5cm]{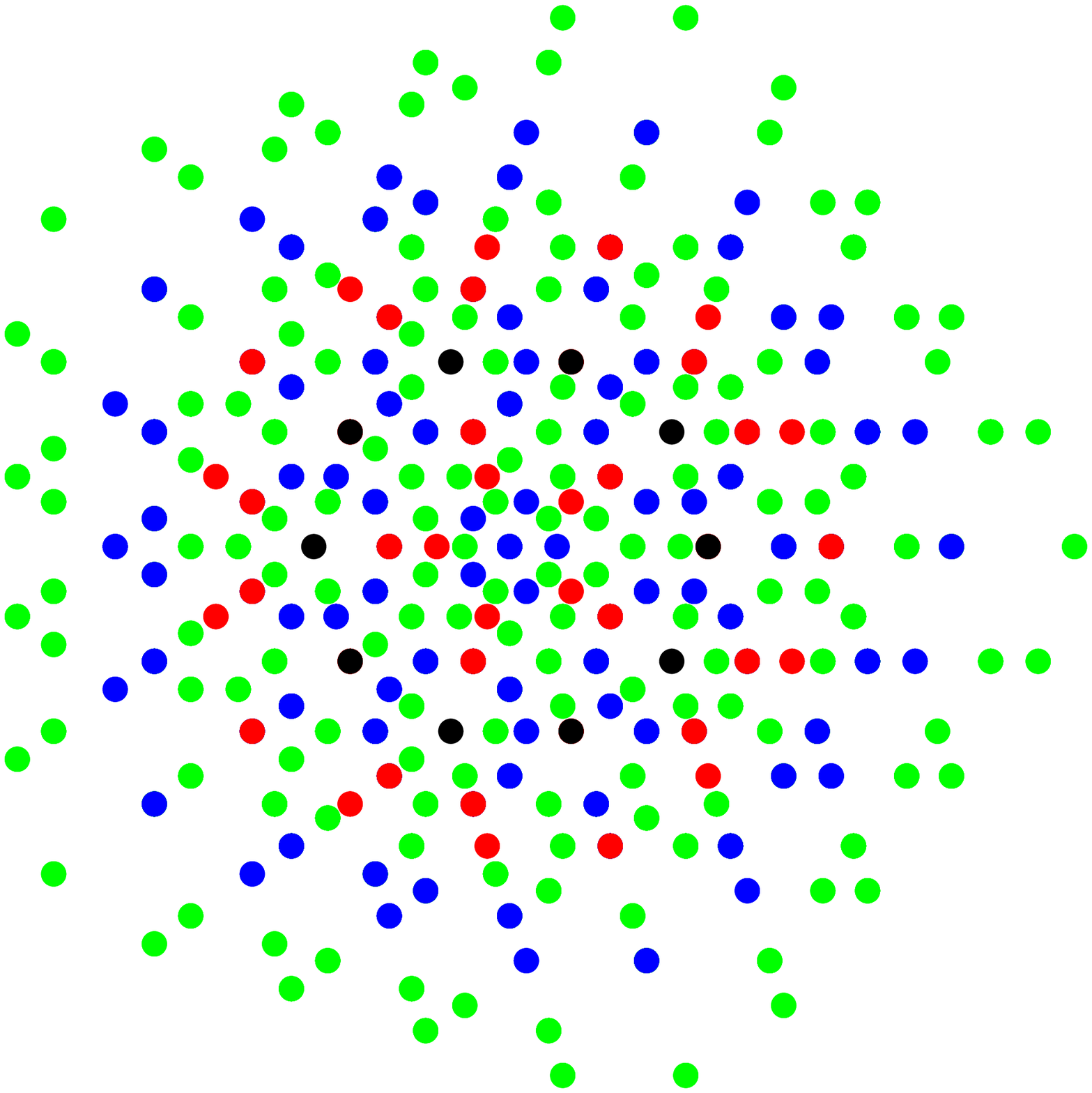}};
					\end{tikzpicture}

		\vspace{0.25cm}
			\\
			(g) $|P_2 (1)|=40$&(h) $|P_2 (2)|=101$&(i) $|P_2 (3)|=206$\\
	 	\vspace{0.25cm}
	  \end{tabular}	

  \caption[Hi]{Quasicrystal fragments for various affine extensions of $H_2$   (the black dots are the decagonal root system): Vertically, the affine root is of length $1$ and parallel to the highest root for the first row of panels (a-c), which are the $H_2^{aff}$-quasicrystal fragments considered in  \cite{Twarock:2002a}. The second row of panels (d-f) are the quasicrystal fragments obtained for the  extension $H_2^{=}$ induced from $A_4^=$ considered here, where the translation length is $\tau$. The third row of panels (g-i) are the quasicrystal fragments obtained for the  extension $\bar{H}_2^{=}$ with translation length $-\sigma$, or alternatively one may think of it in terms of $\left(A\left(H_2^{=}\right)\right)^T$. 
Horizontally, the panels show the point sets  $Q_2 (1)$, $Q_2 (2)$ and $Q_2 (3)$ derived from the root system by letting the translation operator $T$ act once (red dots), twice (blue dots) and three times (green dots). Thus, panels (a), (d) and (g) correspond to the point sets with cardinalities $36$ and $40$ listed in Table \ref{tabFibCard}. 
The cardinalities of the shells $|P_2 (n)|$ are also given. We note that Galois conjugate translations yield the same cardinalities.}
\label{figHaffT}
\end{center}
\end{figure}

Twarock and Patera  \cite{Twarock:2002a} considered simply-laced affine extensions $H_i^{aff}$ of $H_2$, $H_3$ and $H_4$ in the context of quasicrystals in two, three and four dimensions. 
As mentioned above, the affine reflections of those extensions yield translations $T$ of length $1$ along the highest root $\alpha_H$.
Such  $H_2^{aff}$-induced quasicrystal fragments $Q_2 (1)$, $Q_2 (2)$ and $Q_2 (3)$ are depicted in  panels (a-c) in Fig. \ref{figHaffT}.
Quasicrystals are often induced via a cut-and-project method from a projection of the root lattice, much as in the projection framework considered here. 
We thus consider the implications of our induced affine extensions for the quasicrystal setting. 
Our new projection construction yields different affine extensions from the above $H_i^{aff}$s, with translation lengths $\tau$ and $-\sigma$ along the highest root $\alpha_H$. 
These are  new cases with  asymmetric Cartan matrices and would therefore not arise in the  symmetric setting. 
However, following the same construction as in   \cite{Twarock:2002a} but with the different translation lengths $\tau$ and $-\sigma$ results in a similar subset of a vertex set of a quasicrystal.
For $H_2^=$, the resulting quasicrystal fragments $Q_2 (1)$, $Q_2 (2)$ and $Q_2 (3)$  are shown in panels (d-f) in Fig. \ref{figHaffT}, and the corresponding fragments for $\bar{H}_2^=$ are shown in panels (g-i). 
We furthermore give the cardinalities $|P_2 (1)|$, $|P_2 (2)|$ and $|P_2 (3)|$ in each case.
Here, panel (a) corresponds to the point set of cardinality $36$ in Table \ref{tabFibCard}, and panels (d) and (g) correspond to the  entries with cardinality $40$.
We again note that Galois conjugate affine roots yield the same cardinalities, as in the higher-dimensional cases before. 
Our novel construction thus leads to different types of quasicrystalline point arrays.
We will later consider whether the extensions from   \cite{Twarock:2002a} could similarly be induced from a higher-dimensional setting. We will see that they would correspond to Cartan matrices with positive and fractional (c.f. $H_3^<$ and $H_3^>$ in Eqs (\ref{CarH3affAS0}) and (\ref{CarH3affAS1})) off-diagonal entries (Section \ref{sec_concl}), making the case for a suitable generalisation of the standard approach by analysing generalised Cartan matrices over extended number fields.

The above projection procedure has thus yielded asymmetric induced Cartan matrices. 
In the context of Kac-Moody algebras and Coxeter groups, it is often of interest to know if an asymmetric (generalised) Cartan matrix $A$ is  {symmetrisable}:
\begin{defn}[Symmetrisability] An asymmetric Cartan matrix $A$ is  {\emph{symmetrisable}} if there exist a diagonal matrix $D$ with positive integer entries and a symmetric matrix $S$ such that $A=DS$.\label{defsym}
\end{defn} 
%
We have investigated the symmetrisability of the induced non-symmetric Cartan matrices \cite{DechantTwarockBoehm2011H3aff}. They are indeed symmetrisable, but the  entries of the resulting symmetric matrices are no longer from $\mathbb{Z}[\tau]$ (see also the discussion in Section \ref{sec_concl}). Given that the Cartan matrix is defined in terms of the geometry of the roots as $A_{ij}=2(\alpha_i\vert \alpha_j)/(\alpha_i\vert \alpha_i)$, i.e. is given in terms of the angles between root vectors and their length, such matrices would imply a geometry for the root system that is no longer  compatible with an (aperiodic) quasilattice, and the corresponding affine groups would therefore lose their distinctive structure. 
Indeed, it is that relation with quasilattices that makes these affine extended groups mathematically interesting, and distinguishes them from the free group obtained by an extension via a random translation. Therefore, we will not  use these symmetric matrices in our context.

\subsection{Invariance of the projections under Dynkin diagram automorphisms}\label{sec_autos}

Before we classify the induced affine extensions, we show in this section that no additional induced extensions arise from 
 the Dynkin diagram automorphisms of the simple and affine Lie algebras considered above.

\begin{lem}[Invariance of the induced extensions] The induced affine extensions $H_2^=$, $H_3^>$, $H_3^=$, $H_3^<$, $\bar{H}_2^=$, $\bar{H}_3^>$, $\bar{H}_3^=$ and $\bar{H}_3^<$ are invariant under the Dynkin diagram automorphisms of $A_4$, $D_6$, $A_4^=$, $D_6^>$, $D_6^=$ and $D_6^<$.
\end{lem}

\begin{proof} We consider the four cases in turn.
\begin{enumerate} \item The Dynkin diagram of $D_6$ has a $\mathbb{Z}_2$-automorphism that acts by permuting the roots $\alpha_5$ and $\alpha_6$ (denoted as $5\leftrightarrow 6$ in the following). The projection displayed in Fig. \ref{figD6}, however, is not symmetric in $\alpha_5$ and $\alpha_6$. Therefore,  the choice of projection could potentially alter the induced affine extension. 
However, as can be seen from equations (\ref{affroot2}), (\ref{affroot0}) and (\ref{affroot1}), all three possible affine roots are in fact invariant under the exchange of  $a_5$ and $a_6$, so that the result of the projection is not affected.

\item Similarly, the simply-laced extension $D_6^=$ has an additional $\mathcal{D}_4$ automorphism symmetry (here $\mathcal{D}_n$ denotes the dihedral group of order $n$) that allows one to swap the roots labelled by $5\leftrightarrow 6$ or $0\leftrightarrow 1$ separately, as well as an overall left-right symmetry of the diagram obtained by swapping the pairs of terminal roots $(0, 1) \leftrightarrow (5,6)$ together with $2\leftrightarrow 4, 3\leftrightarrow 3$\cite{FuchsSchweigert1997}. 
This symmetry is made manifest in the diagram shown in Fig. \ref{figD6p}. 
Thus, the four terminal roots are equivalent, and one could define four different projections, depending on which terminal root one considered as the affine root. Once one decides on the affine root, the rest of the diagram is fixed by the projection. However, the formula for the affine root is symmetric in $(0, 1, 5, 6)$ as can be seen from  (\ref{affroot2}). Thus,  the induced affine extension is again independent of which projection one chooses.

\condcomment{\boolean{includefigs}}{
\condcomment{\boolean{tikzfigs}}{
\begin{figure}
\begin{center}
\begin{tikzpicture}[
    knoten/.style={        circle,      inner sep=.15cm,        draw}
   ]
  \node at (1,1.5) (knoten0) [knoten,  color=white!0!black] {};  
  \node at (1,-.5) (knoten1) [knoten,  color=white!0!black] {};
  \node at  (3,.5) (knoten2) [knoten,  color=white!0!black] {};
  \node at  (5,.5) (knoten3) [knoten,  color=white!0!black] {};
  \node at  (7,.5) (knoten4) [knoten,  color=white!0!black] {};
  \node at (9,-.5) (knoten5) [knoten,  color=white!0!black] {};
  \node at (9,1.5) (knoten6) [knoten,  color=white!0!black] {};

\node at  (0.5,1.5)  (alpha0) {$\alpha_0$};
\node at  (.5,-.5)  (alpha1) {$\alpha_1$};
\node at  (3,0)  (alpha2) {$\alpha_2$};
\node at  (5,0)  (alpha3) {$\alpha_3$};
\node at  (7,0)  (alpha4) {$\alpha_4$};
\node at  (9.5,-.5)  (alpha5) {$\alpha_5$};
\node at  (9.5,1.5)  (alpha6) {$\alpha_6$};

  \path  (knoten0) edge (knoten2);
  \path  (knoten1) edge (knoten2);
  \path  (knoten2) edge (knoten3);
  \path  (knoten3) edge (knoten4);
  \path  (knoten4) edge (knoten5);
  \path  (knoten4) edge (knoten6);
 
\end{tikzpicture} 
\caption[$D_6$ auto]{A more symmetric version of the $D_6^=$ Dynkin diagram, that makes the $\mathcal{D}_4$-automorphism symmetry manifest.}
\label{figD6p}
\end{center}
\end{figure}}

\condcomment{\boolean{psfigs}}{
\begin{figure}
\begin{center}
	\includegraphics[width=12cm]{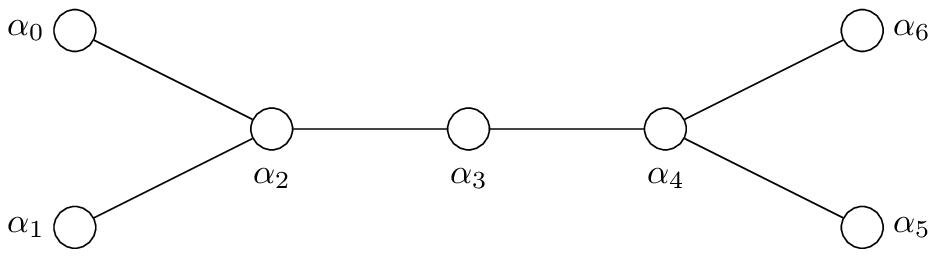}
\caption[$D_6$ auto]{A more symmetric version of the $D_6^=$ Dynkin diagram, that makes the $\mathcal{D}_4$-automorphism symmetry manifest.}
\label{figD6p}
\end{center}
\end{figure}}
}

\item Likewise, the $A_4$ diagram has a $\mathbb{Z}_2$-symmetry swapping left and right, that is broken by the projection. However, the affine root (\ref{affroota4}) is again invariant, so that the induced affine extension is not affected. 

\item The extended diagram $A_4^=$ has an enhanced $\mathcal{D}_5$-automorphism symmetry \cite{FuchsSchweigert1997}, under which the affine root can be seen as  invariant by rewriting Eq. (\ref{affroota4}) as
\begin{equation}
\alpha_0+\alpha_1+\alpha_2+\alpha_3+\alpha_4=0.\label{affroota4sym}
\end{equation}
Thus, in this case, one could choose any of the roots of the extended diagram as the affine root, and the others are then  fixed by the projection. 
\end{enumerate}
The invariance is at the level of the affine roots before projection, so  it does not matter into which invariant subspace one projects. 
Thus, the Dynkin diagram automorphisms do not affect the induced affine extensions.  
\end{proof}

\subsection{Extending by two nodes}\label{sec_further}
Until now, we have considered extending a diagram by a single root, and we have projected this single affine root onto the single induced affine root. 
However, as is shown in Fig. \ref{figE8}, other roots are projected in pairs, e.g. $\alpha_1$ and $\alpha_7$ project onto $a_1$ and $\tau a_1$, which results in the single $H_4$-root $a_1$.
In analogy, we now consider affine extensions of the diagrams by two nodes such that  the two additional roots project as a pair onto a single affine root.
This can be achieved  by further extending the above affine extensions by another node, or by extending the initial diagrams by two nodes at once. 


We first check whether further extending the above groups $A_4^=$, $D_6^>$, $D_6^=$, $D_6^<$ and $E_8^=$ by another node leads to new induced affine extensions. We show that in such a case, the only possibilities are in fact the above diagrams with a disconnected node. Thus, this type of extension is trivial, and the additional affine root will not be a superposition of the other roots: 

\begin{lem} \label{lem1} The affine extensions of $A_4^=$, $D_6^>$, $D_6^=$, $D_6^<$ and $E_8^=$ by a further node are  disconnected.
\end{lem}

\begin{proof} The determinants of general Kac-Moody-type affine extensions of  $A_4^=$, $D_6^>$, $D_6^=$, $D_6^<$ and $E_8^=$ are quadratic in the entries in the new row and column with negative coefficients. Since the entries are non-positive, the determinant is therefore also non-positive. Since the zero entries in a Cartan matrix are symmetric, the determinant vanishes if and only if all the entries in the new row and column vanish. Thus, the extended diagram has a disconnected node.  
\end{proof}

The other possibility is to extend by two nodes at once, and to demand that the Cartan matrix of the double-extension be affine, i.e. that it has zero determinant, but that none of the principal minors has this property. 

\begin{defn}[Affine double extension] An \emph{affine double extension}  is a Kac-Moody-type extension of a diagram by two nodes.
\end{defn}

We analyse here the simply-laced double extensions with a trivial  projection kernel, which give 292 such matrices for $E_8$, 27 for $D_6$, and 6 for $A_4$. Note, however, that the number of different Dynkin diagrams is actually lower. For instance, there are only three diagrams that occur for $A_4^{++}$, which are displayed in Fig. \ref{figA4pp}. The first  corresponds to a single extension of the $D$-series, c.f. $D_6^=$ above. The second is  a single extension of $A_5$, and the third a diagram with a trivial disconnected node. The diagrams for $D_6$ and $E_8$ have a richer mathematical structure. However, it shall suffice here to consider these matrices from the projection point of view -- a more detailed analysis of double extensions will be  relegated to future work.

\begin{lem} \label{lem2} There are no   simply-laced affine double extensions of $A_4$, $D_6$ and $E_8$ with trivial projection kernel. 
\end{lem}

\begin{proof} In all the cases mentioned above (292 for $E_8$,  27 for $D_6$ and 6 for $A_4$), it is not possible to express both additional roots simultaneously  in terms of linear combinations of the roots of the unextended group, c.f. the diagrams for $A_4$ in Fig.  \ref{figA4pp}. Hence, the Cartan matrices can be obtained only in terms of higher-dimensional vectors, i.e. the kernel is non-trivial. 
\end{proof}

\begin{cor} Simply-laced affine double extensions of $A_4$, $D_6$ and $E_8$ with trivial projection kernel do not induce any further affine extensions of the non-crystallographic Coxeter groups. 
\end{cor}

\condcomment{\boolean{includefigs}}{
\condcomment{\boolean{tikzfigs}}{
\begin{figure}
	\begin{center}
       \begin{tabular}{@{}c@{ }c@{ }c@{ }}
		\begin{tikzpicture}[scale=0.50,
		    knoten/.style={        circle,      inner sep=.15cm,        draw}
		   ]
		  \node at (1,1.5) (knoten11) [knoten,  color=white!0!black] {};  
		  \node at (1,-.5) (knoten0) [knoten,  color=white!0!black] {};
		  \node at  (3,.5) (knoten1) [knoten,  color=white!0!black] {};
		  \node at (5,.5) (knoten4) [knoten,  color=white!0!black] {};
		  \node at (7,-.5) (knoten5) [knoten,  color=white!0!black] {};
		  \node at (7,1.5) (knoten6) [knoten,  color=white!0!black] {};
	\node at (9,0.5) {};

		  \path  (knoten0) edge (knoten1);
		  \path  (knoten11) edge (knoten1);
		 \path  (knoten1) edge (knoten4);
		  \path  (knoten4) edge (knoten5);
		  \path  (knoten4) edge (knoten6);
\hspace{2cm}
		\end{tikzpicture}& 
		\begin{tikzpicture}[scale=0.50,
		    knoten/.style={        circle,      inner sep=.15cm,        draw}
		   ]
		  \node at (3,2.0) (knoten11) [knoten,  color=white!0!black] {};
		  \node at (5,2.0) (knoten0) [knoten,  color=white!0!black] {};
		  \node at (1,0.5) (knoten1) [knoten,  color=white!0!black] {};
		  \node at (3,0.5) (knoten2) [knoten,  color=white!0!black] {};
		  \node at (5,0.5) (knoten3) [knoten,  color=white!0!black] {};
		  \node at (7,0.5) (knoten4) [knoten,  color=white!0!black] {};
 		\node at (9,0.5) {};

		  \path  (knoten11) edge (knoten1);
		  \path  (knoten0) edge (knoten11);
		  \path  (knoten0) edge (knoten4);
		  \path  (knoten1) edge (knoten2);
		  \path  (knoten2) edge (knoten3);
		  \path  (knoten3) edge (knoten4);
\hspace{2cm}
		\end{tikzpicture} &
 		\begin{tikzpicture}[scale=0.50,
		    knoten/.style={        circle,      inner sep=.15cm,        draw}
		   ]
		  \node at (7,2.0) (knoten11) [knoten,  color=white!0!black] {};
		  \node at (4,2.0) (knoten0) [knoten,  color=white!0!black] {};
		  \node at (1,0.5) (knoten1) [knoten,  color=white!0!black] {};
		  \node at (3,0.5) (knoten2) [knoten,  color=white!0!black] {};
		  \node at (5,0.5) (knoten3) [knoten,  color=white!0!black] {};
		  \node at (7,0.5) (knoten4) [knoten,  color=white!0!black] {};
 		\node at (9,0.5) {};

		  \path  (knoten0) edge (knoten1);
		  \path  (knoten0) edge (knoten4);
		  \path  (knoten1) edge (knoten2);
		  \path  (knoten2) edge (knoten3);
		  \path  (knoten3) edge (knoten4);
		\end{tikzpicture} 
  \\
  \end{tabular}
  \caption[$A_4^{++}$]{Double extensions $A_4^{++}$ of $A_4$.}
\label{figA4pp}
\end{center}
\end{figure}}

\condcomment{\boolean{psfigs}}{
\begin{figure}
	\begin{center}
		\includegraphics[width=16cm]{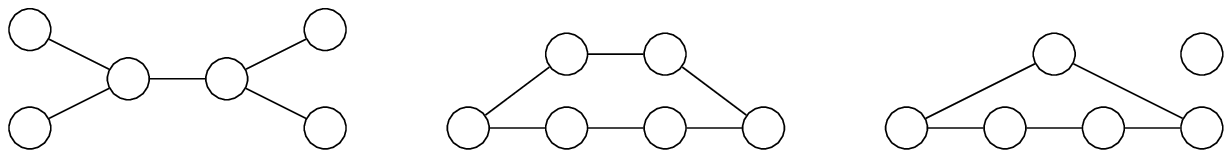}
  \caption[$A_4^{++}$]{Double extensions $A_4^{++}$ of $A_4$.}
\label{figA4pp}
\end{center}
\end{figure}}
}

Earlier, we have demonstrated that the induced affine extensions do not depend on the non-trivial automorphism properties of the simple and extended diagrams. Therefore, in summary, we conclude that the ten cases considered above are actually the only cases that arise in the context of trivial projection kernels. 

\section{Classification of induced affine extensions}\label{class}
The affine extensions induced via the projection in Section \ref{sec_single} (right arrow in Fig. \ref{EDAH_overview}) are subsets of the infinite families developed purely in a non-crystallographic framework in  \cite{DechantTwarockBoehm2011H3aff} (bottom arrow in Fig. \ref{EDAH_overview}). Therefore, we first summarise the relevant results from this paper in Subsection \ref{sec_sum}, and then analyse in Subsection  \ref{sec_class} how  the induced affine extensions  relate to our classification in  \cite{DechantTwarockBoehm2011H3aff}. 

\subsection{Construction of affine extensions in the non-crystallographic case}\label{sec_sum}

In the case of non-crystallographic Coxeter groups, which have Cartan matrices given in terms of  the extended integer ring   $\mathbb{Z}[\tau]$, the earlier definition of affine extensions from Section \ref{sec_std_single}  via  introducing  {affine hyperplanes} $H_{\alpha_0 ,i}$ as solutions to the equations $(x \vert \alpha_0 ) =i\,$, where $x\in \mathcal{E}$, $\alpha_0\in \Phi$ and $i\in\mathbb{Z}$, is not possible because the crystallographic restriction \cite{Senechal:1996}  implies that the planes cannot be stacked periodically; however, $i \in \mathbb{Z}[\tau]$ is too general because $\mathbb{Z}[\tau]$ is dense in $\mathbb{R}$.  

\condcomment{\boolean{includefigs}}{
\condcomment{\boolean{tikzfigs}}{
\begin{figure}

		\begin{tikzpicture}[
			    knoten/.style={        circle,      inner sep=.1cm,        draw}
			   ]
						\node at (0,1.5) (orig){};
			\node at (2+2.0,3) (knoten0) [knoten,  color=white!0!black] {};
			  \node at (1+2.0,1.5) (knoten1) [knoten,  color=white!0!black] {};
			  \node at (3+2.0,1.5) (knoten2) [knoten,  color=white!0!black] {};
			  \node at (1.2+2.0,2.3)  (sigma) {$\frac{5}{2}$};
			  \node at (2.8+2.0,2.3)  (sigma) {$\frac{5}{2}$};
			  \node at (2+2.0,1.75)  (tau) {$5$};
			  \path  (knoten1) edge (knoten2);
			  \path  (knoten1) edge (knoten0);
			  \path  (knoten0) edge (knoten2);
		\node at (10.3,2.2) (CM) 
		{\small{$A \left(H_2^{aff}\right) = \begin{pmatrix} 2&\tau'&\tau' \\ \tau'&2&-\tau\\ \tau'&-\tau&2 \end{pmatrix}$}};
		\end{tikzpicture}\\
		\begin{tikzpicture}[
			    knoten/.style={        circle,      inner sep=.1cm,        draw}
			   ]
			\node at (0,1.5) (orig){};
				  \node at (3+1.0,3) (knoten0) [knoten,  color=white!0!black] {};  
				  \node at (1+1.0,1.5) (knoten1) [knoten,  color=white!0!black] {};
				  \node at (3+1.0,1.5) (knoten2) [knoten,  color=white!0!black] {};
				  \node at (5+1.0,1.5) (knoten3) [knoten,  color=white!0!black] {};
				\node at (4+1.0,1.75)  (tau) {$5$};
				\node at (2.8+1.0,2.3)  (sigma) {$\frac{5}{2}$};
				  \path  (knoten1) edge (knoten2);
				  \path  (knoten2) edge (knoten3);
				  \path  (knoten0) edge (knoten2);
		\node at (10.7,2.1) (CM) 
		{\small{$A \left(H_3^{aff}\right)  = \begin{pmatrix} 2&0&\tau'&0 \\ 0&2&-1&0 \\ \tau'&-1&2&-\tau\\ 0&0&-\tau&2 \end{pmatrix}$}};
		\end{tikzpicture}\\
		\begin{tikzpicture}[
			    knoten/.style={        circle,      inner sep=.1cm,        draw}
			   ]
						\node at (0,0) (orig){};
				 \node at (-1,0) (knoten0) [knoten,  color=white!0!black] {};  
				  \node at (1,0) (knoten1) [knoten,  color=white!0!black] {};
				  \node at (3,0) (knoten2) [knoten,  color=white!0!black] {};
				  \node at (5,0) (knoten3) [knoten,  color=white!0!black] {};
				  \node at (7,0) (knoten4) [knoten,  color=white!0!black] {};

				\node at (0,0.25)  (tau) {$\frac{5}{2}$};
				\node at (6,0.25)  (tau) {$5$};

				\node at (4,1.5)  (tau2) {};

				  \path  (knoten0) edge (knoten1);
				  \path  (knoten1) edge (knoten2);
				  \path  (knoten2) edge (knoten3);
				  \path  (knoten3) edge (knoten4);
		\node at (10,0) (CM) 
		{\small{$A \left(H_4^{aff}\right)= \begin{pmatrix} 2&\tau'&0&0&0 \\ \tau'&2&-1&0&0 \\ 0&-1&2&-1&0 \\ 0&0&-1&2&-\tau \\ 0&0&0&-\tau&2 \end{pmatrix}$}};

		\end{tikzpicture}				
  \caption[Hiaff]{
Coxeter-Dynkin diagrams and Cartan matrices for (from top to bottom) $H_2^{aff}$, $H_3^{aff}$ and $H_4^{aff}$, the unique symmetric affine extensions in  \cite{Twarock:2002a}. 
Note that Coxeter angles of $2\pi/5$ lead to labels $\frac{5}{2}$ (or, $\tau'$ in the notation of  \cite{Twarock:2002a}).
 In  \cite{DechantTwarockBoehm2011H3aff}, we have found infinite families of  generalisations of these examples, which are obtained from the symmetric cases via scalings with $\tau$ and thus follow a Fibonacci recursion relation. We have also found more general examples, which likewise display this scaling property.}
\label{figHiaff}
\end{figure}}

\condcomment{\boolean{psfigs}}{
\begin{figure}
				\includegraphics[width=16cm]{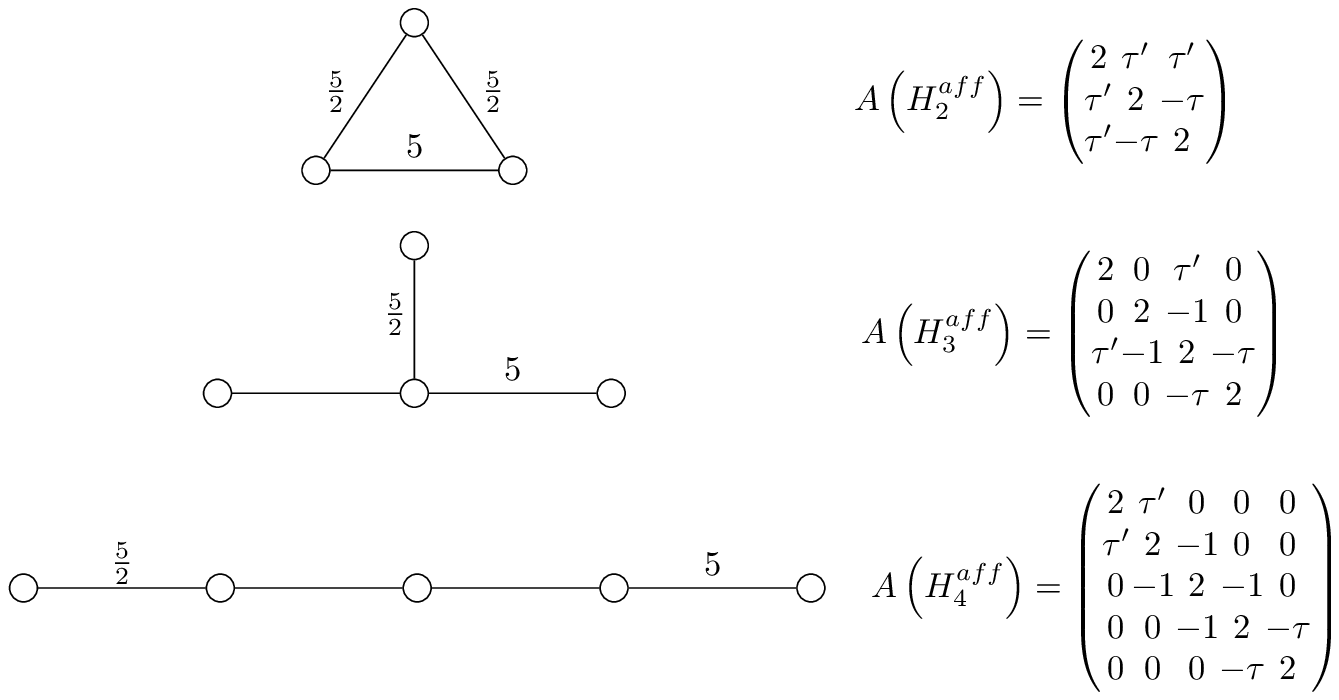}
  \caption[Hiaff]{
Coxeter-Dynkin diagrams and Cartan matrices for (from top to bottom) $H_2^{aff}$, $H_3^{aff}$ and $H_4^{aff}$, the unique symmetric affine extensions in  \cite{Twarock:2002a}. 
Note that Coxeter angles of $2\pi/5$ lead to labels $\frac{5}{2}$ (or, $\tau'$ in the notation of  \cite{Twarock:2002a}).
 In  \cite{DechantTwarockBoehm2011H3aff}, we have found infinite families of  generalisations of these examples, which are obtained from the symmetric cases via scalings with $\tau$ and thus follow a Fibonacci recursion relation. We have also found more general examples, which likewise display this scaling property.}
\label{figHiaff}

\end{figure}}
}

In contrast, the definition of affine extensions of non-crystallographic Coxeter groups via the Kac-Moody-type extensions of their Cartan matrices in Definition \ref{defKMtype} still works. 
For instance,   affine extensions along a 2-fold axis of icosahedral symmetry $T_2=(1,0,0)$ (i.e. along the simple roots) have the following general form \cite{DechantTwarockBoehm2011H3aff}
\begin{equation}
A = \begin{pmatrix} 2&0&x &0 \\ 0&2&-1&0 \\ y&-1&2&-\tau\\ 0&0&-\tau&2 \end{pmatrix},\label{CarH3affASfam2}
\end{equation}
since $T_2$ is orthogonal to two of the simple roots, and thus there is only one pair of off-diagonal entries that is non-zero. 
 This family of matrices contains   $H_3^=$  in  Eq. (\ref{CarH3affAS2}) as a special case.
For this type of matrix  to be affine, the  determinant constraint $\det A=xy-\sigma^2=0$ determines the product  $A_{13}A_{31}=xy$ of the non-zero entries $x$ and $y$  as $xy=2-\tau=\sigma^2$.
From the definition of the Cartan matrix, the products of the off-diagonal entries give the angle of the affine root with the simple roots, such that $xy$ gives the only non-trivial angle of the affine root with the simple roots.
 Corresponding extensions of $H_2$ and $H_4$ satisfy the same constraint. 
This determinant constraint therefore includes the symmetric affine extensions $H_i^{aff}$ ($i=2, 3, 4$) found in  \cite{Twarock:2002a}, which satisfy $x=y=\sigma$ ($\tau'$ in their notation) and are displayed in  Fig.  \ref{figHiaff}.
Writing $x=(a+\tau b)$ and $y=(c+\tau d)$ with $a,b,c,d \in \mathbb{Z}$, and denoting $(x,y)$ by the quadruplet $(a,b;c,d)$, the $H_i^{aff}$  correspond to the simplest case $(a,b;c,d)=(1, -1; 1, -1)$. The  units in $\mathbb{Z}[\tau]$ are the powers of $\tau$, so scaling $x\rightarrow \tau^{-k} x$ and $y\rightarrow \tau^k y$ ($k\in\mathbb{Z}$) leaves the product $xy$ invariant, and thus one can generate a series of solutions to the determinant constraint from any particular reference solution. In terms of quadruplets $(a,b;c,d)$, this scaling of the series of solutions by $(\tau^{-k}, \tau^k)$ amounts  to  $(a,b;c,d)\rightarrow(b-a, a; d, c+d)$, which is a Fibonacci defining relation. The  case  $x=y=\sigma$ is distinguished by its symmetry, and we choose to generate the whole \emph{Fibonacci series of solutions} from this particular reference solution. Thus, under the $\tau$-rescaling $(\tau^{-1}, \tau)$, $H_3^=$ is the `first' asymmetric example with $(x,y)=(-\sigma^2, -1)=(\tau-2, -1)$ corresponding to the quadruplet $(-2, 1;-1, 0)$. Since the powers of $\tau$ are the only units in $\mathbb{Z}[\tau]$, these solutions are in fact the only solutions to the determinant constraint.
The length of a root is given by $\sqrt{x/y}$, so that $\tau$-rescalings do not change the angle but generate affine roots of different lengths.  
There is thus a countably infinite set ($k\in \mathbb{Z}$) of  affine extensions of $H_i$ with affine reflection hyperplanes at  distances  $\tau^{k}/2$ from the origin. For any given $k$, there is  an infinite stack of parallel planes with separation $\tau^{k}/2$, including one containing the origin, which corresponds to one of the  reflections in the unextended group.

By choosing an ansatz similar to Eq. (\ref{CarH3affASfam2}), one can also obtain translations along 3- and 5-fold axes of icosahedral symmetry, $T_3=(\tau, 0, \sigma)$ and $T_5=(\tau, -1,0)$. These cases correspond to one pair of non-zero entries $(x,y)$ in the Cartan matrix  in the fourth row and column, and in the second row and column, respectively, since  these symmetry axes are again orthogonal to two of the simple roots. From now on, we will use the notation 
$A^{aff} = \begin{pmatrix}  2&\underbar{v}^T \\ \underbar{w}&A  \end{pmatrix}$ to write affine extensions succinctly in terms of the  vectors $\underbar{v}$ and $\underbar{w}$ in the additional row and column. Thus, affine extensions along a 3-fold axis correspond to $\underbar{v}=(0,0,x)^T$ and $\underbar{w}=(0,0,y)^T$, whilst affine extensions along a 5-fold axis correspond to $\underbar{v}=(x,0,0)^T$ and $\underbar{w}=(y,0,0)^T$.
These affine extensions lead to determinant constraints that do not have solutions in $\mathbb{Z}[\tau]$, e.g. $xy=\frac{4}{3}\sigma^2$ for an affine extension along a 3-fold axis. If one is prepared to relax the conditions on the Cartan matrix entries  to $\mathbb{Q}[\tau]$ instead of $\mathbb{Z}[\tau]$,  one can solve the determinant constraint  in $\mathbb{Q}[\tau]$, and generate similar families of affine extensions via $\tau$-rescalings from a particular reference solution. 
We introduce a pair of fractions $\gamma, \delta \in\mathbb{Q}$   and write the entries of the Cartan matrix as  $\mathbb{Z}[\tau]$-integers multiplied by $\gamma, \delta$ as $x=\gamma(a+\tau b)$ and $y=\delta(c+\tau d)$. The   solution is therefore now given by a  quadruplet $(a,b;c,d)$ plus multipliers $(\gamma, \delta)$. The latter need to make up the fraction in the determinant constraint, e.g. $\gamma\delta=\frac{4}{3}$ for an affine extension along a 3-fold axis. For instance, $(1,-1;1,-1)$ and $(1, \frac{4}{3})$ correspond to the solution $(\sigma, \frac{4}{3}\sigma)$ of length $\sqrt{x/y}=\frac{1}{2}\sqrt{3}=\frac{1}{2}|T_3|$.

For extensions along the 5-fold axis,  the determinant constraint   is proportional to  $(3-\tau)$, which cannot be solved symmetrically ($x=y$) in $\mathbb{Z}[\tau]$; one solution is, for instance, $(1,-2;1,-1)$. This implies that swapping $x$ and  $y$ produces a different solution  to the determinant constraint. Since the length of the affine root is given by $\sqrt{{x}/{y}}$, this generates a solution of different length. These two solutions are independent and  generate two Fibonacci families of  affine roots by $\tau$-multiplication. For the previous determinant constraints with symmetric solutions,  transposition and $\tau$-rescalings are equivalent, such that only one Fibonacci family arises.

In summary, one can thus label an affine extension in the following way: A solution is given in terms of an integer quadruplet $(a,b;c,d)$ that is related to a particular reference solution via rescaling with a power $k$ of $\tau$, together with a multiplier pair $(\gamma, \delta)$.

\subsection{Identification of the induced affine extensions in the Fibonacci classification}\label{sec_class}
 
We now identify the induced affine extensions derived in Section \ref{sec_single} within the Fibonacci families from  \cite{DechantTwarockBoehm2011H3aff} discussed in the previous Section \ref{sec_sum}.

The  affine extensions $H_i^{=}$ induced from the three simply-laced affine extensions of the crystallographic groups are all related to the $H_i^{aff}$s via $\tau$-rescalings. In particular, we have seen that $H_i^{=}$ corresponds to $(-2, 1;-1, 0)$, which is derived from the symmetric solution $(1,-1;1,-1)$ corresponding to $H_i^{aff}$ via rescaling $(x,y)$ as $(\tau^{-1}x, \tau y)$. Likewise, $\left(A(H_i^=)\right)^T$ corresponds to $(\tau x, \tau^{-1} y)$ and is equivalent to $\bar{H}_i^=$. The induced affine extensions found in this paper are therefore the `first' asymmetric solutions in  the Fibonacci family with the symmetric reference solution $H_i^{aff}$ obtained by rescaling with one power of $\tau$. Since for these examples, the determinant constraint $xy=\sigma^2$ is solved in $\mathbb{Z}[\tau]$, the multiplier pair $(\gamma, \delta)$ is trivially $(1,1)$. 

The determinant constraint for the other two induced affine extensions of $H_3$ found above, $H_3^<$ and $H_3^>$,  is    $xy=\frac{4}{5}(3-\tau)$, which has no symmetric solution. Thus,  two inequivalent series of solutions are generated by the quadruplets $(-1,0;-3,1)$ and $(-3,1;-1,0)$ or, equivalently, from their $\tau$-multiples $(1,-1;1,-2)$ and $(1,-2;1,-1)$. 
In our notation, the multiplier pair $(\gamma, \delta)=(\frac{4}{5}, 1)$ and the quadruplet $(-3,1;-1,0)$ give  $H_3^<$, and $(\gamma, \delta)=(\frac{2}{5}, 2)$ with $(-3,1;-1,0)$ gives $H_3^>$, such that they both belong to a Fibonacci family found in \cite{DechantTwarockBoehm2011H3aff}, represented by  $(1,-2;1,-1)$ and scaled by $(\tau^{-1}, \tau)$.   

We summarise the relation of all induced affine extensions to the Fibonacci classification \cite{DechantTwarockBoehm2011H3aff} in Table \ref{tabFibClassAll}.

\begin{thm}[Classification] The only  affine extensions of the non-crystallographic groups $H_4$, $H_3$ and $H_2$ induced via projection from at most simply-laced double extensions of $E_8$, $D_6$ and $A_4$ with trivial projection kernel are those in Table \ref{tabFibClassAll}, and  they are a subset of the Fibonacci Classification scheme presented in \cite{DechantTwarockBoehm2011H3aff}.
\end{thm}

\begin{proof} The five known affine extensions of $E_8$, $D_6$ and $A_4$ give the induced extensions listed in the table. By the Invariance Lemma, the projection is invariant under the (extended) Dynkin diagram automorphisms. Further extensions by a single node would be disconnected, and the simply-laced double extensions are incompatible with the projection formalism via the two Lemmas \ref{lem1} and \ref{lem2} in Section \ref{sec_further}. Thus, no more cases arise in this setting. The classification was performed in \cite{DechantTwarockBoehm2011H3aff}.
\end{proof}

\begin{table}
\begin{centering}\begin{tabular}{|c||c||c|c|c||c|c|}
\hline
group&$xy$&$(a,b;c,d)_{ref}$&$k$&$(\gamma, \delta)$&$\underbar{v}^T$&$\underbar{w}^T$
\tabularnewline
\hline
\hline
$H_4^=$&$2-\tau$&$(1, -1; 1, -1)$&$-1$&$(1,1)$&$(\tau-2, 0, 0,0)$&$(-1,0,0,0)$
\tabularnewline
\hline
$H_3^=$&$2-\tau$&$(1, -1; 1, -1)$&$-1$&$(1,1)$&$(0, \tau-2, 0)$&$(0,-1,0)$
\tabularnewline
\hline
$H_3^<$&$\frac{4}{5}(3-\tau)$&$(1,-2;1,-1)$&$-1$&$(\frac{4}{5}, 1)$&$(\frac{4}{5}(\tau-3), 0, 0)$&$(-1, 0, 0)$
\tabularnewline
\hline
$H_3^>$&$\frac{4}{5}(3-\tau)$&$(1,-2;1,-1)$&$-1$&$(\frac{2}{5}, 2)$&$(\frac{2}{5}(\tau-3), 0, 0)$&$(-2, 0, 0)$
\tabularnewline
\hline
$H_2^=$&$2-\tau$&$(1, -1; 1, -1)$&$-1$&$(1,1)$&$(\tau-2, \tau-2)$&$(-1,-1)$
\tabularnewline
\hline
\end{tabular}\par\end{centering}
\caption[indclass]{\label{tabFibClassAll} Identification of the induced extensions of $H_2$, $H_3$ and $H_4$ within the Fibonacci classification: For $x=\gamma (a+\tau b), \,\, y=\delta (c+\tau d)$, the $\mathbb{Z}[\tau]$-quadruplet part $(a,b;c,d)$ of the solution $(x,y)$ to the determinant constraint $xy$ (second column) is given by scaling a representative reference solution within a Fibonacci family, e.g. one distinguished by its symmetry (see Section \ref{sec_sum}) like $(1,-1;1,-1)$, (third column) by a power $\tau^k$ of $\tau$ (fourth column). The rational part of the solution is contained in the multiplier pair $(\gamma, \delta)$ (fifth column). This contains all the information to construct   the row and column vectors $\underbar{v}$ and $\underbar{w}$ in the extended Cartan matrix for  $H_4^=$, $H_3^=$, $H_3^<$, $H_3^>$ and $H_2^=$, which are given in the  last two columns. The affine extensions induced by projection into the other invariant subspace  $\bar{H}_4^=$, $\bar{H}_3^=$, $\bar{H}_3^<$, $\bar{H}_3^>$ and $\bar{H}_2^=$ have  Cartan matrices that are Galois conjugate to the above five cases and that follow a similar Fibonacci classification in terms of $\sigma=1-\tau$. They are essentially equivalent to  transposes of the former type since the two sets of simple roots $a_i$ and $\bar{a}_i$ generate equivalent compact groups $H_i$. }
\end{table}

\section{Discussion}\label{sec_concl}

We have shown that via  affine extensions of the crystallographic root systems $E_8$, $D_6$ and $A_4$ (upper arrow in Fig. \ref{EDAH_overview}) and subsequent projection (dashed arrow), one obtains  affine extensions of the non-crystallographic groups $H_4$, $H_3$ and $H_2$ of the type considered in  \cite{DechantTwarockBoehm2011H3aff} (lower arrow). 
This provides an alternative construction of affine extensions of this type, and by placing them into the broader context of the crystallographic group $E_8$, we open up new potential applications in Lie theory, modular form theory and high energy physics. 
The ten induced extensions derived here  are a subset of the extensions in the Fibonacci classification scheme derived in  \cite{DechantTwarockBoehm2011H3aff}. The Fibonacci classification contains an infinity  of solutions of which the ones derived here  are thus a subset distinguished by the projection. 
For the simply-laced cases, the induced extensions $H_i^=$ and $\bar{H}_i^=$ (since this is equivalent to $A(H_i^=)^T$)
can be derived from  the symmetric solutions $H_i^{aff}$ from \cite{Twarock:2002a} via rescaling with $\tau$ 
and are in that sense the `first' asymmetric members of the corresponding Fibonacci families of solutions. These distinguished affine extensions could thus have a special r\^ole in practical applications, e.g. in quasicrystal theory, virology and carbon chemistry. 

The induced extensions are  $\mathbb{Z}[\tau]$-valued in the simply-laced cases, and $\mathbb{Q}[\tau]$-valued  for the other two non-simply-laced cases. 
This  suggests to  further generalise the Kac-Moody framework of \cite{DechantTwarockBoehm2011H3aff} to
allow extended number fields in the entries in the extended Cartan matrices of $H_2$, $H_3$ and $H_4$; this could be $\mathbb{Q}[\tau]$, but a milder extension might also suffice. 
One could therefore argue to also allow corresponding generalisations in the extended Cartan matrices of  $E_8$, $D_6$ and $A_4$, from which the non-crystallographic cases are obtained via projection.  
Such a generalisation might lead to interesting mathematical structures and could open up novel applications in hyperbolic geometry and rational conformal field theory, where similar fractional values can occur \cite{Petkova2000BCsinRCFT, Petkova1998ClassBBCFT, Coxeter1973regular}. Various other approaches hint at this same generalisation, which we now explore in turn. 

In particular, the projection (left arrow in Fig. \ref{EDAH_overview}) is in fact one-to-one, since integer Cartan matrix entries in the higher-dimensional setting project onto $\mathbb{Z}[\tau]$-integers in half the number of dimensions, and the two parts in a $\mathbb{Z}[\tau]$-integer do not mix. 
This is due to the irrationality of the projection angle, which projects a lattice in higher dimensions to an aperiodic quasilattice in lower dimensions, without a null space of the same dimension.
Thus, one can invert the projection by  `lifting' the affine roots and thereby Cartan matrices of the non-crystallographic groups considered in the Fibonacci classification in  \cite{DechantTwarockBoehm2011H3aff} to those of the crystallographic groups (i.e. by inverting the dashed arrow). 
We now lift (denoted by $L$)  such extended Cartan matrices of $H_4$, $H_3$ and $H_2$ in order to determine what type of  extensions of  $E_8$, $D_6$ and $A_4$  could  induce  them via projection (denoted by $P$).  Again, we denote generic extensions $E_8^+$, $D_6^+$ and $A_4^+$ by their additional row and column vector in the Cartan matrix as follows

$$A \left(E_8^+\right) = \begin{pmatrix}
   2&\underbar{v}^T
\\ \underbar{w}&E_8
 \end{pmatrix}, \,\,\,
A \left(D_6^+\right) = \begin{pmatrix}
   2&\underbar{v}^T
\\ \underbar{w}&D_6
 \end{pmatrix} \text{ and }
A \left(A_4^+\right) = \begin{pmatrix}
   2&\underbar{v}^T
\\ \underbar{w}&A_4
 \end{pmatrix}.$$

We begin by lifting the affine extensions  $H_i^{aff}$ from  \cite{Twarock:2002a}, which are the symmetric special cases in the Fibonacci families of solutions. One might have thought intuitively that these $H_i^{aff}$ would arise via projection, rather than the $H_i^{=}$. It is thus interesting to lift the affine roots of the $H_i^{aff}$ to the higher-dimensional setting and to see which Cartan matrix they would therefore give rise to. For example, for $H_4^{aff}$ the vectors giving the additional row and column in the Cartan matrix are given by
\begin{equation}
L\left(A\left(H_4^{aff}\right)\right)=	 \begin{pmatrix}
	   2&\underbar{v}^T
	\\ \underbar{w}&E_8
	 \end{pmatrix} \text{ with } \underbar{v}^{aff}_4=\underbar{w}^{aff}_4=(1,0,0,0,0,0,-1,0)^T\label{H4afflift}.
\end{equation}
Similarly, the vectors  corresponding to $H_3^{aff}$ and $H_2^{aff}$ are  given by
$\underbar{v}^{aff}_3=\underbar{w}^{aff}_3=(0,1,0,-1,0,0)^T$ and
$\underbar{v}^{aff}_2=\underbar{w}^{aff}_2=(1,-1,-1,1)^T$, respectively.
We note that the lifted versions of the symmetric extensions $H_i^{aff}$ are also symmetric.
However, the requirement of non-positivity of the off-diagonal Cartan matrix entries that is usual in the Lie algebra context is not satisfied by these matrices. This is in agreement with the fact that the only standard affine extensions of   $E_8$, $D_6$ and $A_4$ are the five cases presented in Section \ref{sec_std_single}. The lifted versions of  $H_i^{aff}$ could thus motivate to relax this requirement of non-positivity in order to arrive at an interesting more general class of Cartan matrices.  


Analogously, one can consider  lifting  the transposes of the Cartan matrices obtained earlier  that are induced from $\pi_\parallel$ (c.f. equations (\ref{CarH4affAS})-(\ref{CarH2affAS})), e.g.  $\left(A\left(H_4^=\right)\right)^T$. In particular,  they are also contained in the Fibonacci classification of affine extensions  in  \cite{DechantTwarockBoehm2011H3aff} (see Section \ref{sec_sum}). They are in fact also related to the extensions induced by $\pi_\perp$, since they give rise to equivalent compact parts with the same translation lengths. 
For example, lifting  the transpose $\left(A\left(H_4^=\right)\right)^T$ of $A\left(H_4^=\right)$ in Eq. (\ref{CarH4affAS}) gives the following matrix in 9D
\begin{equation}
LTP\left(A\left(E_8^=\right)\right)=L\left(\left(A\left(H_4^=\right)\right)^T\right)=	 \begin{pmatrix}
	   2&\underbar{v}^T
	\\ \underbar{w}&E_8
	 \end{pmatrix} \text{ with } \underbar{v}^{=}_4=\frac{1}{2}\underbar{w}^{=}_4=(-1,0,0,0,0,0,\frac{1}{2},0)^T\label{CarH4affAST},
\end{equation}
%
where we have denoted the combination of projecting the affine extension of $E_8$, transposing, and lifting again,  by $LTP$. 
We  note that there are again positive, but now also fractional entries -- neither occurs in the context of  simple Lie  Theory. We also observe that the consistency conditions (the Lemma in  \cite{DechantTwarockBoehm2011H3aff}) stipulated in our previous paper are still obeyed. 
%
%
One may find rational entries surprising, but these actually arise naturally in the context of the affine extensions considered in the non-crystallographic setting \cite{DechantTwarockBoehm2011H3aff}, e.g. $H_3^<$. Perhaps, therefore generalising integer to rational  entries also in the higher-dimensional crystallographic case could lead to interesting new mathematical structures. 
The other cases correponding to $D_6^<$, $D_6^>$, $D_6^=$ and $A_4^=$ are given by
$\underbar{v}^<_3=\frac{5}{12}\underbar{w}^<_3=(-1,0,0,0,0,\frac{1}{3})^T$, 
$\underbar{v}^>_3=\frac{5}{3}\underbar{w}^>_3=(-2,0,0,0,0,\frac{2}{3})^T$,
$\underbar{v}^=_3=\frac{1}{2}\underbar{w}^=_3=(0,-1,0,\frac{1}{2},0,0)^T$, and
$\underbar{v}^=_2=\frac{1}{2}\underbar{w}^=_2=(-1,\frac{1}{2},\frac{1}{2},-1)^T$, respectively. 
In an analogous manner, one could proceed to lift all the solutions in the Fibonacci family rather than just $H_i^{aff}$ and $A\left(H_i^{=}\right)^T$, but these instructive examples shall suffice to give some indication towards the  generalisations that arise.

 
%
%
%
As explained in Definition \ref{defsym} in Section \ref{sec_proj_aff}, in the Coxeter group and Lie Algebra contexts, one is often interested in symmetrisable Cartan matrices. 
For completeness, we therefore present here the symmetrised version of $LTP\left(A\left(E_8^=\right)\right)$, which we denote by $SLTP$, as an example of which type of matrix arises through symmetrisation
\begin{equation}
SLTP\left(A\left(E_8^=\right)\right)=	\begin{pmatrix}
	   1&\underbar{v}^T
	\\ \underbar{w}&E_8\end{pmatrix} \text{ with } \underbar{v}^T=\underbar{w}^T=(-1,0,0,0,0,0,\frac{1}{2},0)\label{E8liftsym}.
\end{equation}
Here we have relaxed the requirement that the symmetrised matrix be integer-valued, since even before symmetrisation, the Cartan matrix has fractional values. Again, positive and fractional values arise. This matrix is  positive semi-definite, which was expected as that corresponds to the affine case. 
Thus, even the criterion of symmetrisability suggests positive entries and extending the number field for the Cartan matrix entries to $\mathbb{Q}[\tau]$ or $\mathbb{Z}[\tau]+\frac{1}{2}\mathbb{Z}[\tau]$.

In summary, we have provided a novel, alternative construction of affine extensions of the type considered in \cite{DechantTwarockBoehm2011H3aff} from the two familiar concepts of affine extensions of crystallographic groups and projection of root systems. 
This construction results in a special subset of the point arrays used in mathematical virology and carbon chemistry that is distinguished via the projection, which could therefore play a special r\^ole in applications.
It also extends the quasicrystal framework considered in  \cite{Twarock:2002a} to a  wider class of quasicrystals. 
We furthermore made the case for admitting extended number fields in the Cartan matrix.
These extended number fields arise in a variety of cases, namely in the lower-dimensional picture \cite{DechantTwarockBoehm2011H3aff},  via projection, via lifting the Fibonacci family of solutions (including $H_i^{aff}$ from \cite{Twarock:2002a} and transposes of $H_i^{=}$) and via symmetrising.   
Fractional entries in  Cartan matrices  arise in hyperbolic geometry and rational conformal field theory.
Our construction here is thus another example of such fractional entries that could open up a new type of analysis and enticing possibilities in these fields.

\section*{Acknowledgements}
We would like to thank the anonymous referee and Prof Vladimir Dobrev for a careful reading of the manuscript and many valuable suggestions. RT gratefully acknowledges support via a Research Leadership Award from the Leverhulme Trust that has provided funding for PPD.



\end{document}